\documentclass[letterpaper,12pt]{article}

\usepackage{amsmath}
\usepackage{amssymb}
\usepackage{graphicx}
\usepackage{url}
\usepackage{natbib}
\usepackage{multirow}
\usepackage{lscape}
\usepackage{amsthm}
\usepackage[margin=1in]{geometry}
\usepackage{mathtools}
\usepackage{tikz}
\usepackage{authblk}
\usepackage{subcaption}
\usepackage{ifthen}
\usepackage{placeins}
\usepackage{longtable}
\usepackage{booktabs}

\theoremstyle{plain}
\newtheorem{theorem}{Theorem}[section]
\newtheorem{corollary}[theorem]{Corollary}
\newtheorem{proposition}[theorem]{Proposition}

\theoremstyle{remark}
\newtheorem{remark}[theorem]{Remark}

\usetikzlibrary{shapes,decorations,arrows,calc,arrows.meta,fit,positioning}

\renewcommand\hat{\widehat}

\newcommand{\doo}{\mathrm{do}}

\newcommand{\Pa}{\mathrm{pa}}

\newcommand{\rec}[1][]{%
  \ifthenelse{ \equal{#1}{} }
    {\mathrm{re}}
    {\mathrm{re}_{#1}}
}
\newcommand{\emi}[1][]{%
  \ifthenelse{ \equal{#1}{} }
    {\mathrm{em}}
    {\mathrm{em}_{#1}}
}

\newcommand{\An}[1][]{%
  \ifthenelse{ \equal{#1}{} }
    {\mathrm{an}}
    {\mathrm{an}_{#1}}
}
\newcommand{\De}[1][]{%
  \ifthenelse{ \equal{#1}{} }
    {\mathrm{de}}
    {\mathrm{de}_{#1}}
}

\newcommand{\V}{\mathrm{V}}
\newcommand{\E}{E}

\newcommand{\logit}{\mathrm{logit}}

\begin{document}

\tikzset{
    -Latex,auto,node distance =1 cm and 1 cm,semithick,
    state/.style ={circle, draw, minimum width = 0.7 cm},
    box/.style ={rectangle, draw, minimum width = 0.7 cm, fill=lightgray},
    point/.style = {circle, draw, inner sep=0.08cm,fill,node contents={}},
    bidirected/.style={Latex-Latex,dashed},
    el/.style = {inner sep=3pt, align=left, sloped}
}

\title{Graph-based causal variance decompositions: When ``variance explained'' means causation}

\author[1]{Olli Saarela\thanks{Correspondence to: Olli Saarela, Dalla Lana School of Public Health, University of Toronto, 155 College St, Toronto, ON, Canada M5T 3M7. Email: olli.saarela@utoronto.ca}}
\author[2]{Juha Karvanen}

\affil[1]{University of Toronto}
\affil[2]{University of Jyv\"askyl\"a}

\date{\today}

\maketitle

\begin{abstract}
Recursive application of the law of total variance decomposes the marginal variance of an outcome into components attributed to explanatory variables and a residual component. The resulting decomposition depends on the chosen conditioning order, and its components do not in general have causal interpretations. We develop a graph-based framework for defining causal counterparts of ordered variance components and establishing their identification from observed data. Under topological orderings, identification can be assessed component by component against the full causal graph, without requiring the variables included in the decomposition to form a causally sufficient system. We also consider scientifically motivated departures from topological orderings, in which selected intermediate variables are conditioned on to obtain controlled-effect interpretations, motivated by the context of disparities in healthcare delivery. We relate the resulting estimands to causal attribution and variable-importance approaches in machine learning, and propose model-based plug-in estimators together with an approximate Bayesian procedure for uncertainty quantification. A simulation study examines finite-sample performance and sensitivity to outcome-model misspecification and flexible machine-learning estimation.
\end{abstract}

\noindent {\small \textbf{Keywords:} Causal attribution, Causal identification, Causal variance decomposition, Variable importance, Directed acyclic graphs}

\section{Introduction}

Understanding the sources of variation in an outcome is a central objective of statistical analysis. Variance decompositions are among the oldest tools for quantifying ``explained variation'', but in multivariable settings they are neither unique nor automatically causal. In general, once more than one source of variation is involved, a multiway decomposition depends on the order of conditioning, and different orders encode different scientific questions. This point was made explicitly by \citet{bowsher2012identifying}, who developed a general sequential-conditioning decomposition for stochastic biochemical networks and emphasized both the non-uniqueness of multiway variance decompositions and the substantive role of the chosen conditioning order.

In the context of institutional quality-of-care comparisons in healthcare, \citet{chen2020causal} proposed a three-way causal variance decomposition for patient-level outcomes, partitioning the observed marginal outcome variance into a component explained by patient case-mix, a component causally explained by hospital performance conditional on case-mix, and a residual component. Their key contribution was to define the decomposition directly in terms of potential outcomes, so that the hospital-related variance component is a population-level causal estimand. The framework was subsequently extended to mediation \citep{chen2022causal}, decomposing between-hospital variation into contributions operating through a process-of-care mediator and through other pathways, and to hierarchical provider structures \citep{chen2023hierarchical}, separating between-hospital and within-hospital between-provider variation. Although these decompositions were developed for a particular applied setting rather than from a general graph-theoretic construction, their successive conditioning can be viewed as following a topological ordering of the variables in the corresponding causal DAGs.

Following the same line of work, \citet{yu2025causal} developed a causal variance decomposition for detecting group disparities in hospital profiling. Their eight-way decomposition separates contributions related to group and hospital effects, effect modification of the hospital effect by group membership, hospital access or selection, case-mix, and residual variation. Motivated by disparity questions that condition on ``allowable'' characteristics, they considered a modified conditioning order in which case-mix covariates precede group membership, and further decomposed the resulting terms into pathway- and effect-modification-specific contributions. Here conditioning on the allowable covariates on the causal pathway gives a ``controlled'' interpretation for the resulting group effects, and gives an example where changing the order of conditioning changes the causal question represented by the variance decomposition.

A connection to variable importance in machine learning was pointed out by \citet{khan2025marginal}. Rather than constructing a sequential decomposition by a causal ordering, their conditional variable importance metric for a predictor $X_j$ conditions on all remaining explanatory variables $X_{-j}$. When $X_j$ is interpreted as a treatment or exposure and $X_{-j}$ constitutes a valid confounder adjustment set, the metric can be represented in terms of conditional average treatment effects. For a categorical exposure, the hospital-related component in the causal variance decomposition of \citet{chen2020causal} is one half of this conditional variable importance metric. However, because \citet{khan2025marginal} defined importance of each predictor relative to the full set of the remaining predictors, they cannot obtain simultaneous causal interpretation for the importance measures of all the variables. A central contribution of \citet{khan2025marginal} was to highlight how marginal variable-importance measures may require extrapolation beyond well-supported covariate regions, and to propose conditional variable importance as a solution, closely connecting this problem to positivity violations in causal inference. Related recent work has developed variable-importance measures specifically for treatment-effect heterogeneity, together with formal inferential procedures. \citet{hines2025variable} propose nonparametric treatment-effect variable-importance measures with influence-function-based estimation and inference, while \citet{paillard2025measuring} develop a conditional permutation approach for statistically assessing the importance of variables in explaining heterogeneity in conditional average treatment effects.

Recent work in machine learning has pursued a more general notion of causal attribution. \citet{janzing2024quantifying} introduced intrinsic causal contributions (ICCs), defined in structural causal models by resolving each observed node into its exogenous noise contribution and interpreting the resulting attribution through structure-preserving interventions. For a fixed ordering, a variance-based ICC is a reduction in expected conditional variance when one additional node-specific noise variable is included in the conditioning set. Their principal construction removes dependence on the ordering through Shapley symmetrization over all permutations. They also consider an asymmetric variant obtained by restricting the averaging to topological orderings. For a fixed topological ordering, the ICC can be expressed using the observed variables that precede each node, rather than the underlying exogenous noise variables, which connects ICC to terms in a sequential observed-data variance decomposition. The main difference is that ICC attributes all uncertainty in the target to the nodes in the causal system, whereas the Chen et al. style decomposition retains a residual component representing outcome variation not explained by the variables included in the decomposition.

\citet{saha2025measuring} developed the topological-order ICC further for global attribution in deep neural networks. They note that a DAG may admit several topological orderings and define a topological ICC by averaging the attribution over all valid topological orderings. Under an assumption of no latent confounding, they use the observed-variable representation of topological ICC to establish identifiability and develop a generative post-hoc estimation framework. Their work therefore provides an identifiable general construction for intrinsic causal attribution under topological orderings, while retaining Shapley-style averaging to resolve non-uniqueness of the ordering.
\citet{jung2022measuring} proposed do-Shapley values, which decompose an intervention-specific expected outcome into contributions of individual causes, and developed graphical identification conditions and robust estimators from observational data. These targets differ from the population-level variance components considered here, but similarly separate the definition of causal contributions from their identification and estimation.

A related line of work develops causal analogues of functional analysis of variance (ANOVA) and global sensitivity analysis. \citet{gao2024counterfactual} propose ``counterfactual explainability'', using a DAG and counterfactual comparisons to extend ANOVA-type attribution to dependent explanatory variables. \citet{zhang2026causal} subsequently develop semiparametric estimation and inference for causal explainability quantities, including influence-function-based estimation and inference for components and interactions. These approaches further demonstrate that explained variation can be formulated as a causal estimand, although their counterfactual ANOVA targets and identifying assumptions differ from the sequential conditional-variance components considered here. More generally, graph-based causal reasoning is now routine \citep{pearl2009causality}, with practical software support for DAG representation and adjustment analysis available \citep[e.g.][]{dagitty,Tikka:identifying}.

Taken together, this literature leaves a gap between application-specific causal variance decompositions and general causal attribution frameworks. The decompositions of \citet{chen2020causal,chen2022causal,chen2023hierarchical} and \citet{yu2025causal} establish causally interpretable variance components for particular DAGs and scientific questions, but do not provide general graphical conditions under which an arbitrary component of an ordered observed-data variance decomposition has a causal interpretation. The conditional variable importance measure of \citet{khan2025marginal} does not produce a sequential multi-component decomposition. Conversely, \citet{janzing2024quantifying,saha2025measuring} provide general graph-based attribution constructions, but focus on intrinsic node-wise contributions and resolve non-unique orderings through averaging. Existing results also typically formulate the topological case under unconfounded or causally sufficient models, whereas in applications the variables included in the decomposition may be only an observed subset of a larger causal system containing unmeasured variables. This motivates a framework in which causal interpretation and identification are assessed component by component against the full graph, and in which alternative scientifically chosen conditioning orders can themselves define substantively different estimands.

We address this gap by developing a general graph-based theory for ordered variance decompositions of stochastic outcomes. We first characterize the causal counterparts of components arising from topological orderings and give graphical conditions under which these causal components are identified by the corresponding observed-data variance components. Causal sufficiency with respect to the observed variables is not assumed: when the decomposition includes only a subset of the variables in a larger causal system, identification can be assessed separately for each component using the full causal graph. We then consider scientifically motivated departures from the topological order and characterize the resulting controlled-effect causal components, giving graphical conditions for their identification that distinguish preceding variables that must be jointly controlled from those that can serve as an adjustment set. This is motivated by disparity or fairness related applications where controlling for descendant or intermediate variables is relevant to the substantive question \citep{karvanen2024simulating,yu2025causal}. Also, rather than averaging over alternative admissible orderings, we treat the ordering and the granularity of the nodes as part of the definition of the scientific estimand, using vector-valued clusters when substantively parallel variables are more naturally attributed jointly. 

Finally, we develop plug-in estimation and approximate Bayesian uncertainty quantification for these decompositions and investigate their finite-sample behavior, including sensitivity to outcome-model misspecification and flexible machine-learning estimation.

\section{Ordered variance decomposition}

\subsection{Notation and ordered variance decomposition}

Let $X_1 \prec X_2 \prec \cdots \prec X_k \prec Y$ denote an ordering of the explanatory random variables $X_1,\ldots,X_k$ and the outcome $Y$. Write
$X_{1:j}:=(X_1,\ldots,X_j)$, $X_{1:0}:=\varnothing$, and $X_{\prec j}:=X_{1:j-1}$ for the variables preceding $X_j$ in the chosen ordering. The ordering within $X_{1:k}$ can be chosen while $Y$ always comes last in the order. Applying the law of total variance sequentially along this ordering gives
\begin{align}
V(Y)
&=
\sum_{j=1}^k
\E_{X_{\prec j}}
\left[
V_{X_j\mid X_{\prec j}}
\left\{
\E(Y\mid X_{1:j})
\right\}
\right]
+
\E_{X_{1:k}}
\left[
V(Y\mid X_{1:k})
\right].
\label{eq:ordered_variance_decomposition}
\end{align}
For $j=1$, the expectation over the empty preceding set is understood to be
the identity operator.

We denote the component attributed to $X_j$ under the chosen ordering by
\[
\Delta_j
=
\E_{X_{\prec j}}
\left[
V_{X_j\mid X_{\prec j}}
\left\{
\E(Y\mid X_{1:j})
\right\}
\right],
\]
and the residual component by
\[
\Delta_{\mathrm{res}}
=
\E_{X_{1:k}}
\left[
V(Y\mid X_{1:k})
\right].
\]
Thus,
\[
V(Y)
=
\sum_{j=1}^k\Delta_j+\Delta_{\mathrm{res}}.
\]

This decomposition is a probabilistic identity and is valid for any specified
ordering of the explanatory variables. The choice of ordering determines how
the explained variation is allocated among the components; a causal
interpretation of an individual component requires additional assumptions,
which we consider in Section~\ref{sec:causal_identification}.

\subsection{Examples}

As a running example that will be expanded later, consider the DAG in Figure~\ref{fig:dag}, adapted from \citet{naimi2016mediation}, where $Z$ is a categorical sociodemographic group indicator, $A$ is a categorical hospital assignment indicator, $\boldsymbol X$ is a vector of clinical case-mix covariates, $Y$ is a process-type quality-of-care outcome (e.g. timely treatment delivery) and $U$ is a possibly present unobserved common cause of $\boldsymbol X$ and $Y$. We treat $\boldsymbol X$ as a vector-valued cluster node \citep{tikka2023clustering}. Baseline covariates could be included but we omit them to simplify notation.

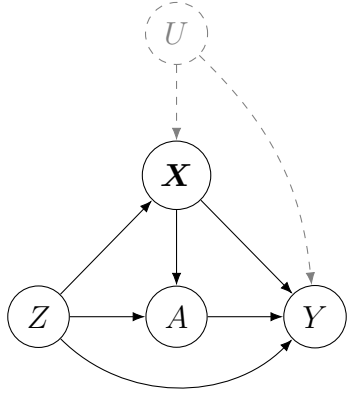
\begin{figure}[h]
\begin{center}
\begin{tikzpicture}
    \node[state] (z) at (0,0) {$Z$};
    \node[state] (a) [right =of z] {$A$};
    \node[state] (y) [right =of a] {$Y$};
    \node[state] (x) [above =of a] {$\boldsymbol X$};

    \node[state, draw=gray, text=gray, dashed] (u)
        [above =of x] {$U$};

    \path (z) edge (x);
    \path (z) edge (a);
    \path (z) edge[bend right=45] (y);
    \path (x) edge (a);
    \path (x) edge (y);
    \path (a) edge (y);

    \path[gray, dashed] (u) edge (x);
    \path[gray, dashed] (u) edge[bend left=20] (y);
\end{tikzpicture}
\end{center}
\caption{Running example with sociodemographic group $Z$, case-mix
covariates $\boldsymbol X$, hospital assignment $A$, and outcome $Y$.
The gray dashed node $U$ represents an unobserved common cause of
$\boldsymbol X$ and $Y$. When $U$ is absent, the observed-variable DAG
is causally sufficient.}
\label{fig:dag}
\end{figure}

We consider first the special case where $U$ is absent. The ordering $Z \prec \boldsymbol X \prec A \prec Y$ is a topological ordering of the resulting DAG. The corresponding ordered variance decomposition is
\begin{align}
V(Y)
&=
V_Z\left\{\E(Y\mid Z)\right\}
\nonumber\\
&\quad+
\E_Z\left[
V_{\boldsymbol X\mid Z}
\left\{\E(Y\mid Z,\boldsymbol X)\right\}
\right]
\nonumber\\
&\quad+
\E_{Z,\boldsymbol X}\left[
V_{A\mid Z,\boldsymbol X}
\left\{\E(Y\mid Z,\boldsymbol X,A)\right\}
\right]
\nonumber\\
&\quad+
\E_{Z,\boldsymbol X,A}
\left[
V(Y\mid Z,\boldsymbol X,A)
\right].
\label{eq:example_ordered_decomposition_simplified}
\end{align}
The first three terms are the variance components attributed to $Z$,
$\boldsymbol X$, and $A$, respectively, under this ordering, and the final
term is the residual variance after conditioning on $(Z,\boldsymbol X,A)$.

For the same DAG, consider instead the modified-order $\boldsymbol X \prec Z \prec A \prec Y$. The corresponding ordered variance decomposition is
\begin{align}
V(Y)
&=
V_{\boldsymbol X}
\left\{\E(Y\mid\boldsymbol X)\right\}
\nonumber\\
&\quad+
\E_{\boldsymbol X}\left[
V_{Z\mid\boldsymbol X}
\left\{\E(Y\mid\boldsymbol X,Z)\right\}
\right]
\nonumber\\
&\quad+
\E_{\boldsymbol X,Z}\left[
V_{A\mid\boldsymbol X,Z}
\left\{\E(Y\mid\boldsymbol X,Z,A)\right\}
\right]
\nonumber\\
&\quad+
\E_{\boldsymbol X,Z,A}
\left[
V(Y\mid\boldsymbol X,Z,A)
\right].
\label{eq:example_modified_ordered_decomposition_simplified}
\end{align}
Here the first three terms are attributed to $\boldsymbol X$, $Z$, and $A$,
respectively. Although this ordering is not topological, the variance
decomposition remains a valid probabilistic identity. Whether its individual
components admit causal interpretations is a separate identification question,
considered in the following section.

\section{Causal identification}\label{sec:causal_identification}

\subsection{Causal identification under topological ordering}

The ordered variance decomposition in
\eqref{eq:ordered_variance_decomposition} shows that the $j$th variance
component attributable to $X_j$ is determined by the inner conditional mean
$\E(Y\mid X_j,X_{\prec j})$, where $X_{\prec j}$ is the conditioning set.
The following proposition connects this to the corresponding causal quantity.
In what follows, we call $S$ a valid adjustment set for the total effect of a possibly multivariate exposure $E$ on $Y$ if adjustment for $S$ identifies the
interventional distribution under $\doo(E)$ according to the graphical adjustment criterion in $G$. 
For a univariate exposure $X$, a sufficient condition is the usual back-door
criterion: $S$ contains no descendants of $X$ and blocks every back-door
path from $X$ to $Y$.

\begin{proposition}[Componentwise identification under a topological ordering]
\label{prop:topological_identification}

Let $G$ be a causal DAG that may contain variables not included in the
variance decomposition, including unobserved variables. Let
$X_1,\ldots,X_k,Y$ be observed variables, ordered as
$X_1\prec\cdots\prec X_k\prec Y$, with this ordering compatible with a
topological ordering of these variables in $G$.

For any $j=1,\ldots,k$, if $X_{\prec j}$ is a valid adjustment set for
the total effect of $X_j$ on $Y$ in the full causal graph $G$, then
\[
\E(Y\mid X_j,X_{\prec j})
=
\E(Y\mid \doo(X_j),X_{\prec j}),
\]
and hence
\[
\Delta_j
=
\E_{X_{\prec j}}
\left[
V_{X_j\mid X_{\prec j}}
\left\{
\E(Y\mid X_j,X_{\prec j})
\right\}
\right]
=
\E_{X_{\prec j}}
\left[
V_{X_j\mid X_{\prec j}}
\left\{
\E(Y\mid \doo(X_j),X_{\prec j})
\right\}
\right].
\]
Thus, whenever the adjustment condition holds for component $j$, its
causal counterpart is identified by the corresponding observed-data
variance component.
\end{proposition}

\begin{proof}[Proof]
Fix $j\in\{1,\ldots,k\}$. By assumption, $X_{\prec j}$ is a valid
adjustment set for the total effect of $X_j$ on $Y$ in the full causal
graph $G$. Hence $Y\perp X_j\mid X_{\prec j}$ in $G_{\underline{X_j}}$,
where $G_{\underline{X_j}}$ denotes the graph obtained by deleting arrows
emanating from $X_j$. By Rule 2 of do-calculus (action/observation
exchange; \citealp{pearl2009causality}),
$p(y\mid \doo(x_j),x_{\prec j}) = p(y\mid x_j,x_{\prec j})$,
and therefore $\E(Y\mid \doo(X_j),X_{\prec j}) = \E(Y\mid X_j,X_{\prec j})$.
Substituting this identity into the observed-data variance component gives
the identification result for $\Delta_j$.
\end{proof}

Whether $X_{\prec j}$ is a valid adjustment set can be checked graphically for each component, for example using \texttt{isAdjustmentSet} function in \texttt{dagitty} package \citep{dagitty} (examples to follow). In the special case where $G$ is causally sufficient, no check is needed because of the following.

\begin{corollary}[Causally sufficient special case]
\label{cor:topological_causal_sufficiency}

If $G$ is a causally sufficient DAG with nodes
$X_1,\ldots,X_k,Y$, and
$X_1\prec\cdots\prec X_k\prec Y$ is a topological ordering, then, for every
$j=1,\ldots,k$, $X_{\prec j}$ is a valid adjustment set for the total
effect of $X_j$ on $Y$. Consequently, the adjustment condition in
Proposition~\ref{prop:topological_identification} holds for every
non-residual component, and each causal variance component is identified by
the corresponding observed-data variance component.
\end{corollary}

\begin{proof}[Proof]
Fix $j\in\{1,\ldots,k\}$. Since the ordering is topological, every variable
in $X_{\prec j}=X_{1:j-1}$ is a non-descendant of $X_j$, and every parent
of $X_j$ is contained in $X_{\prec j}$. Any path from $X_j$ to $Y$ that
begins with an arrow into $X_j$ therefore passes through a parent of $X_j$,
which is a noncollider on that path and is contained in $X_{\prec j}$.
Consequently, conditioning on $X_{\prec j}$ blocks every back-door path from
$X_j$ to $Y$. Thus $X_{\prec j}$ is a valid adjustment set for the total effect of
$X_j$ on $Y$, and the result follows from Proposition~\ref{prop:topological_identification}.
\end{proof}

\begin{remark}[Relation to parent-based adjustment in the causally sufficient special case]
\label{rem:topological_parents}
If the ordering is topological, every parent of $X_j$ precedes $X_j$, and
therefore $\Pa(X_j) \subseteq X_{\prec j}$.
Thus the antecedent set $X_{\prec j}$ always contains the parent set of $X_j$, which is
a sufficient adjustment set for the total effect of $X_j$ on $Y$.
In some cases, $X_{\prec j} = \Pa(X_j)$,
so that the corollary reduces to the familiar parent-adjusted identity
$\E(Y \mid X_j,\Pa(X_j)) = \E(Y \mid \doo(X_j), \Pa(X_j))$.
\end{remark}

Topological order is not always unique, for example in parallel structures, the simplest example being the DAG $X_1 \to Y \leftarrow X_2$. For such graphs, we propose the following.

\begin{proposition}[Unique topological ordering by clustering]
\label{prop:topological_clustering}

Let $X_1,\ldots,X_k$ be the variables included in the ordered variance
decomposition. Define a partial order $\preceq_G$ on these variables by
$X_i \preceq_G X_j$ if $X_i=X_j$ or $X_i$ is an ancestor of $X_j$ in the full causal graph $G$.
Assume also that $Y$ can be placed after all $X_i$ in a topological ordering of these variables.

If the variables can be partitioned into antichains
\[
\mathcal C_1,\ldots,\mathcal C_K,
\qquad
\bigcup_{\ell=1}^K \mathcal C_\ell
=
\{X_1,\ldots,X_k\},
\]
with $\mathcal C_\ell\cap\mathcal C_m=\varnothing$ for $\ell\neq m$,
such that $X_i \preceq_G X_j$ for every $X_i\in\mathcal C_\ell$ and $X_j\in\mathcal C_m$ whenever $\ell<m$,
then, treating each $\mathcal C_\ell$ as a vector-valued node, the clusters have the unique
topological ordering
$\mathcal C_1
\prec
\mathcal C_2
\prec
\cdots
\prec
\mathcal C_K
\prec Y$.
All topological orderings of the individual variables differ only by
permutations of variables within the clusters. Consequently, the ordered
variance decomposition is uniquely defined at the cluster level.
\end{proposition}

\begin{proof}[Proof]
Consider any two clusters $\mathcal C_\ell$ and $\mathcal C_m$ with
$\ell<m$. By assumption, every variable in $\mathcal C_\ell$ is an
ancestor of every variable in $\mathcal C_m$ in $G$. Hence every
topological ordering of the individual variables must place every member
of $\mathcal C_\ell$ before every member of $\mathcal C_m$. Therefore
the relative ordering of the clusters is necessarily
$\mathcal C_1\prec\mathcal C_2\prec\cdots\prec\mathcal C_K$.

Within any cluster $\mathcal C_\ell$, the variables form an antichain
under $\preceq_G$, so no two distinct members of the cluster are related
by ancestry. Consequently, the ancestral partial order imposes no
ordering constraint among variables within the same cluster, and
different topological orderings can differ only by permutations within
the clusters.

Thus, after treating each $\mathcal C_\ell$ as a vector-valued node, all
within-cluster permutations represent the same cluster-level ordering,
which is uniquely $\mathcal C_1\prec\mathcal C_2\prec\cdots\prec\mathcal C_K\prec Y$.
Since an ordered variance decomposition depends only on the ordering of
the nodes to which components are assigned, the resulting decomposition
is uniquely defined at the cluster level.
\end{proof}

\begin{remark}[Clustering versus averaging]
When several topological orderings remain because variables are causally
parallel, one possible approach is to average the resulting decompositions
over admissible orderings, as in topological-order extensions of
Shapley-based attribution. We do not pursue this approach here. When the
parallel variables have a common substantive role, we instead prefer to
combine them into a vector-valued node. This retains variation arising from
interactions among the variables within the same component rather than
allocating it according to an arbitrary within-cluster ordering.

\end{remark}

\subsection{Example: application of the topological-order identification rule}\label{sec:topologicalexample}

In Figure~\ref{fig:dag}, the observed variables admit the ordering 
$Z\prec\boldsymbol X\prec A\prec Y$.
When the unobserved variable $U$ is present, however, the causal
interpretation of the resulting components must be assessed separately.
The $Z$ component is identified without adjustment, whereas the
$\boldsymbol X$ component is not identified because the back-door path
$\boldsymbol X\leftarrow U\to Y$
remains open after conditioning on $Z$. In contrast, the $A$ component
is identified after conditioning on $(Z,\boldsymbol X)$, which blocks
the back-door paths from $A$ to $Y$. Thus, the presence of unobserved
confounding need not invalidate the causal interpretation of every
component of the decomposition.
This can be verified graphically by the \texttt{dagitty} package functions as
\begin{verbatim}
library(dagitty)

g <- dagitty("dag {
  Z -> X -> A -> Y
  Z -> A
  Z -> Y
  X -> Y
  U -> X
  U -> Y
}")
plot(g)

# Z-term: empty adjustment set for effect of Z on Y
isAdjustmentSet(g, Z = character(0), exposure = "Z", outcome = "Y")

# X-term: adjust for Z for effect of X on Y
isAdjustmentSet(g, Z = "Z", exposure = "X", outcome = "Y")

# A-term: adjust for Z and X for effect of A on Y
isAdjustmentSet(g, Z = c("Z","X"), exposure = "A", outcome = "Y")
\end{verbatim}

If $U$ is absent, the graph is causally sufficient and
$Z\prec\boldsymbol X\prec A\prec Y$ is a topological ordering of all
explanatory variables. In that special case, the preceding variables
form valid adjustment sets automatically, as stated in
Corollary~\ref{cor:topological_causal_sufficiency}.
Hence \eqref{eq:example_ordered_decomposition_simplified} may be written as
\begin{align}
V(Y)
&=
V_Z\bigl(\E(Y \mid \doo(Z))\bigr)
+
\E_Z
\Bigl(
V_{\boldsymbol X \mid Z}
\bigl(\E(Y \mid \doo(\boldsymbol X), Z)\bigr)
\Bigr)
\nonumber\\
&\qquad
+
\E_{\boldsymbol X, Z}
\Bigl(
V_{A \mid Z,\boldsymbol X}
\bigl(\E(Y \mid \doo(A), Z,\boldsymbol X)\bigr)
\Bigr)
\nonumber\\
&\qquad
+
\E_{\boldsymbol X, Z, A}
\bigl[V(Y \mid \boldsymbol X, Z, A)\bigr].
\label{eq:topological_example_causal}
\end{align}

\subsection{Causal identification under modified orderings}

We now consider a user-provided ordering of observed variables
\[
\sigma:\quad
X_{\sigma(1)}
\prec X_{\sigma(2)}
\prec\cdots\prec
X_{\sigma(h)}
\prec Y,
\]
which need not be topological. The variables included in the ordering may be
only a subset of the observed variables in the causal system. In addition,
the full causal DAG $G$ may contain unobserved variables. Variables not
included in the ordering do not receive separate variance components; rather,
the residual component is defined after conditioning only on the variables
included in $\sigma$.

The ordered variance decomposition remains a probabilistic identity for any
such ordering, but its causal interpretation is not automatic. In particular,
some variables preceding $X_j$ under $\sigma$ may be descendants of $X_j$ in
the full causal graph. Such variables cannot be interpreted as ordinary
adjustment variables and are instead treated as variables held fixed when
interpreting the $X_j$ component.

For a variable $X_j$ included in $\sigma$, let
$X_{\prec j}^{\sigma}$ denote the variables that precede $X_j$ under the
user-provided ordering. Define
\[
C_j^\sigma
:=
X_{\prec j}^{\sigma}\cap\De_G(X_j),
\]
where $\De_G(X_j)$ denotes the descendants of $X_j$ in the full causal graph
$G$, including descendants connected to $X_j$ through paths containing
unobserved variables. Let
\[
B_j^\sigma
:=
X_{\prec j}^{\sigma}\setminus C_j^\sigma.
\]
Thus, $C_j^\sigma$ contains the preceding variables that are downstream of
$X_j$ and are interpreted as controlled, whereas $B_j^\sigma$ contains the
remaining preceding variables and serves as the candidate adjustment set.
Define the corresponding joint exposure set as
\[
E_j^\sigma
:=
\{X_j\}\cup C_j^\sigma.
\]
The following proposition addresses causal identification of the modified-order variance components.

\begin{proposition}[Componentwise identification under a modified ordering]
\label{prop:modified_order_identification}

For any variable $X_j$ in the ordering $\sigma$, if $B_j^\sigma$ is a
valid adjustment set for the total effect of the joint exposure set
$E_j^\sigma$ on $Y$ in the full causal graph $G$, then
\[
\E(Y\mid X_j,X_{\prec j}^{\sigma})
=
\E(Y\mid \doo(E_j^\sigma),B_j^\sigma),
\]
and hence
\[
\Delta_j^\sigma
=
\E_{X_{\prec j}^{\sigma}}
\left[
V_{X_j\mid X_{\prec j}^{\sigma}}
\left\{
\E(Y\mid X_j,X_{\prec j}^{\sigma})
\right\}
\right]
=
\E_{X_{\prec j}^{\sigma}}
\left[
V_{X_j\mid X_{\prec j}^{\sigma}}
\left\{
\E(Y\mid \doo(E_j^\sigma),B_j^\sigma)
\right\}
\right].
\]
Thus, whenever the adjustment condition holds for component $j$, its
controlled-effect causal counterpart is identified by the corresponding
observed-data variance component.
\end{proposition}

\begin{proof}[Proof]
Since
\[
X_{\prec j}^{\sigma}
=
B_j^\sigma\cup C_j^\sigma,
\]
we have
\[
\E(Y\mid X_j,X_{\prec j}^{\sigma})
=
\E(Y\mid E_j^\sigma,B_j^\sigma).
\]
If $B_j^\sigma$ is a valid adjustment set for the joint exposure
$E_j^\sigma$, rule 2 of do-calculus gives
\[
\E(Y\mid E_j^\sigma,B_j^\sigma)
=
\E(Y\mid \doo(E_j^\sigma),B_j^\sigma).
\]
The result follows by substitution into the $j$th variance component.
\end{proof}

The condition in Proposition~\ref{prop:modified_order_identification} can 
again be checked graphically for each component. For example, the
\texttt{isAdjustmentSet()} function in \texttt{dagitty} can be applied with
$E_j^\sigma$ as the possibly multivariate exposure, $B_j^\sigma$ as the
candidate adjustment set, and $Y$ as the outcome. Importantly, this check is
performed in the full causal graph, including any unobserved variables.

\begin{remark}[Relation to the topological-order case]
\label{rem:modified_reduces_topological}
If the user-provided ordering is topological, no descendant of $X_j$ can
precede $X_j$, so that $C_j^\sigma=\varnothing$, $B_j^\sigma=X_{\prec j}^{\sigma}$ and $E_j^\sigma=\{X_j\}$.
Proposition~\ref{prop:modified_order_identification} then reduces to the
Proposition~\ref{prop:topological_identification} setting. Under 
the additional conditions of Corollary~\ref{cor:topological_causal_sufficiency}, 
the required adjustment conditions hold automatically.
\end{remark}

\begin{remark}[Graphical failures under a modified ordering]
Because $B_j^\sigma$ is determined by the user-provided ordering rather than
chosen specifically as an adjustment set, it need not satisfy the adjustment
criterion. For example, preceding variables may include colliders or may fail
to block back-door paths involving observed or unobserved common causes.
Unobserved variables require no separate argument: they are retained in the
full causal graph and may cause the adjustment condition to fail.
\end{remark}

\subsection{Example: application of the modified-order identification rule}

Consider again the DAG in Figure~\ref{fig:dag}, but now use the modified
ordering $\boldsymbol X\prec Z\prec A\prec Y$.
This ordering is not topological because $\boldsymbol X$ is a descendant of
$Z$ but precedes $Z$ in the ordering. The motivation for such an
order is to consider between-group disparities in the outcome within
strata of allowable clinical covariates \citep{yu2025causal}. In the context
of the running example, we consider $\boldsymbol X$ to be both treatment (hospital)
and outcome (a process of care) allowable \citep{jackson2021meaningful}.

When the unobserved variable $U$ is present, the causal interpretation of the
components can again be assessed separately. For the $\boldsymbol X$
component there are no preceding variables, so
$C_{\boldsymbol X}^{\sigma} = B_{\boldsymbol X}^{\sigma} = \varnothing$.
Identification would therefore require the empty set to be a valid adjustment
set for the effect of $\boldsymbol X$ on $Y$. This condition fails because
the back-door paths $\boldsymbol X\leftarrow Z\to Y$ and
$\boldsymbol X\leftarrow U\to Y$ remain open.

For the $Z$ component, the preceding variable $\boldsymbol X$ is a descendant
of $Z$, and hence
$C_Z^\sigma=\{\boldsymbol X\}$ and $B_Z^\sigma=\varnothing$.
Thus, $\boldsymbol X$ is treated as controlled and the relevant joint exposure
set is $\{Z,\boldsymbol X\}$. When $U$ is present, however, the empty set is
not a valid adjustment set for the joint effect of $(Z,\boldsymbol X)$ on
$Y$, because the path $\boldsymbol X\leftarrow U\to Y$ remains open. 
The $Z$ component is therefore not causally identified.

For the $A$ component, neither of its preceding variables is a descendant of
$A$, so $C_A^\sigma=\varnothing$ and $B_A^\sigma=\{Z,\boldsymbol X\}$.
Conditioning on $(Z,\boldsymbol X)$ blocks the back-door paths from $A$ to
$Y$, including those involving $U$. Consequently, the $A$ component remains
causally identified.

These conclusions can be verified graphically using the \texttt{dagitty}
package (same \texttt{g} graph as in Section \ref{sec:topologicalexample} example):
\begin{verbatim}
# X-term: empty adjustment set for effect of X on Y
isAdjustmentSet(g, Z = character(0), exposure = "X", outcome = "Y")

# Z-term: X is controlled, giving joint exposure (Z,X)
isAdjustmentSet(g, Z = character(0),
                exposure = c("Z","X"), outcome = "Y")

# A-term: adjust for Z and X for effect of A on Y
isAdjustmentSet(g, Z = c("Z","X"), exposure = "A", outcome = "Y")
\end{verbatim}

If $U$ is absent, the $\boldsymbol X$ component remains unidentified because
the back-door path $\boldsymbol X\leftarrow Z\to Y$
is still open. However, the situation changes for the $Z$ component. Once $U$ is
removed, the empty set is a valid adjustment set for the joint exposure
$(Z,\boldsymbol X)$, and hence
\[
\E(Y\mid Z,\boldsymbol X)
=
\E\{Y\mid\doo(Z,\boldsymbol X)\}.
\]
The $Z$ component therefore has a controlled-effect interpretation in which
$\boldsymbol X$, a descendant of $Z$, is held fixed. The $A$ component
continues to be identified after adjustment for $(Z,\boldsymbol X)$. Thus, in the absence of $U$,
\eqref{eq:example_modified_ordered_decomposition_simplified} may be written as
\begin{align}
V(Y)
&=
V_{\boldsymbol X}
\left\{
\E(Y\mid\boldsymbol X)
\right\}
\nonumber\\
&\qquad+
\E_{\boldsymbol X}
\left[
V_{Z\mid\boldsymbol X}
\left\{
\E(Y\mid\doo(Z,\boldsymbol X))
\right\}
\right]
\nonumber\\
&\qquad+
\E_{\boldsymbol X,Z}
\left[
V_{A\mid\boldsymbol X,Z}
\left\{
\E(Y\mid\doo(A),\boldsymbol X,Z)
\right\}
\right]
\nonumber\\
&\qquad+
\E_{\boldsymbol X,Z,A}
\left[
V(Y\mid\boldsymbol X,Z,A)
\right].
\label{eq:modified_example_causal}
\end{align}
The first component remains associational, whereas the $Z$ and $A$
components have causal interpretations. In particular, the $Z$ component
describes variation in the outcome under interventions on $Z$ while
$\boldsymbol X$ is held fixed, rather than variation due to the total effect
of $Z$. Also, by the law of total variance, the latter 3 terms constitute a decomposition 
of $E_{\boldsymbol X} \left\{ \V(Y\mid\boldsymbol X) \right\}$.

\subsection{Example: using clustering to obtain an identified coarser component}

Clustering can also be used to change the granularity of the decomposition
when a scientifically meaningful component is not identified under a finer
partition of the variables. To illustrate this, consider the extension of
the running example in Figure~\ref{fig:allowable1}. Suppose the case-mix
variables are partitioned into two vector-valued nodes:
$\boldsymbol X$, representing allowable clinical characteristics, and
$\boldsymbol W$, representing non-allowable demographic or socioeconomic
characteristics.

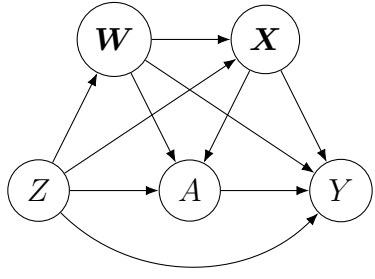
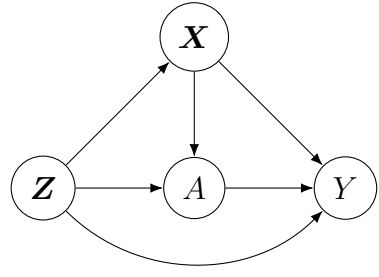
\begin{figure}[h]
\begin{subfigure}{0.5\textwidth}
\begin{center}
\begin{tikzpicture}
    \node[state] (z) at (0,0) {$Z$};
    \node[state] (a) at (2,0) {$A$};
    \node[state] (y) at (4,0) {$Y$};
    \node[state] (w) at (1,2) {$\boldsymbol W$};
    \node[state] (x) at (3,2) {$\boldsymbol X$};

	\path (z) edge (w);
    \path (z) edge (x);
    \path (z) edge (a);
    \path (z) edge[bend right=45] (y);

	\path (w) edge (x);
	\path (w) edge (a);
	\path (w) edge (y);

    \path (x) edge (a);
    \path (x) edge (y);
    \path (a) edge (y);

\end{tikzpicture}
\end{center}
\caption{DAG with clusters $\boldsymbol W$ and $\boldsymbol X$}
\label{fig:allowable1}
\end{subfigure}
\begin{subfigure}{0.5\textwidth}
\begin{center}
\begin{tikzpicture}
    \node[state] (z) at (0,0) {$\boldsymbol Z$};
    \node[state] (a) at (2,0) {$A$};
    \node[state] (y) at (4,0) {$Y$};
    \node[state] (x) at (2,2) {$\boldsymbol X$};

    \path (z) edge (x);
    \path (z) edge (a);
    \path (z) edge[bend right=45] (y);

    \path (x) edge (a);
    \path (x) edge (y);
    \path (a) edge (y);
\end{tikzpicture}
\end{center}
\caption{DAG where $\boldsymbol Z = (Z, \boldsymbol W)$}
\label{fig:allowable2}
\end{subfigure}

\caption{A variant of the hypothetical hospital comparison setting where the variables $\boldsymbol W$ are considered to be non-allowable for adjustment.}
\label{fig:allowable}
\end{figure}

Consider the modified ordering
$\boldsymbol X \prec Z \prec \boldsymbol W \prec A \prec Y$.
For the $Z$ component, $\boldsymbol X$ precedes $Z$ despite being a
descendant of $Z$. Thus, $C_Z^\sigma=\{\boldsymbol X\}$ and $B_Z^\sigma=\varnothing$,
and identification would require
\[
\E(Y\mid \boldsymbol X,Z)
=
\E\{Y\mid\doo(\boldsymbol X,Z)\}.
\]
This equality does not hold in general because the path
$\boldsymbol X\leftarrow\boldsymbol W\to Y$ remains open. 
Hence the $Z$ component under this modified ordering is not causally identified.

Suppose instead that $Z$ and $\boldsymbol W$ are combined into the
vector-valued node $\boldsymbol Z=(Z,\boldsymbol W)$,
as shown in Figure~\ref{fig:allowable2}, and consider the modified ordering
$\boldsymbol X\prec\boldsymbol Z\prec A\prec Y$.
For the $\boldsymbol Z$ component, $\boldsymbol X$ is a preceding descendant, so that
$C_{\boldsymbol Z}^\sigma=\{\boldsymbol X\}$ and $B_{\boldsymbol Z}^\sigma=\varnothing$.
The corresponding joint exposure is $(\boldsymbol Z,\boldsymbol X)=(Z,\boldsymbol W,\boldsymbol X)$.
In the graph of Figure~\ref{fig:allowable2}, the empty set is a valid
adjustment set for this joint exposure, and therefore
\[
\E(Y\mid\boldsymbol X,\boldsymbol Z)
=
\E\{Y\mid\doo(\boldsymbol X,\boldsymbol Z)\}.
\]
The resulting cluster-level component is therefore causally identified.

Importantly, clustering does not identify the original $Z$ component.
Rather, it changes the estimand by attributing variation jointly to
$(Z,\boldsymbol W)$. This may be scientifically appropriate when the
variables in the cluster are intended to represent a common class of
pathways, such as sources of variation via different demographic or socioeconomic factors.

\subsection{Connection to intrinsic causal contributions}

The ordered causal variance decomposition considered here is closely related
to the intrinsic causal contribution (ICC) framework of
\citet{janzing2024quantifying}, although the two frameworks address different
attribution questions. In the basic unconfounded formulation of
\citet{janzing2024quantifying}, the causal system is represented by a
nonparametric structural equation model with independent errors (NPSEM-IE),
\[
X_j = f_j\{\Pa(X_j),N_j\},
\]
where the exogenous noise variables $N_1,\ldots,N_n$ are jointly
independent. They also extend the ICC framework to confounded systems by
allowing dependence among the exogenous noise variables.
Recursively substituting the structural equations
expresses a target variable $Y$ as a function of these exogenous noise
variables. For an index set $T$, let $N_T$ denote the corresponding
subvector of noise variables. Adapting their notation to the present setting,
the intrinsic causal contribution of $X_j$ conditional on $N_T$ is defined
as the reduction in a conditional uncertainty measure $\psi$ obtained by
additionally conditioning on $N_j$:
\[
\mathrm{ICC}_{\psi}(X_j\to Y\mid T)
=
\psi(Y\mid N_T)
-
\psi(Y\mid N_T,N_j).
\]
For the variance-based uncertainty measure
$\psi(Y\mid N_T) = \E\{V(Y\mid N_T)\}$,
this becomes
\[
\mathrm{ICC}_{V}(X_j\to Y\mid T)
=
\E\{V(Y\mid N_T)\}
-
\E\{V(Y\mid N_T,N_j)\}.
\]
Thus, ICC attributes to a node the reduction in expected residual variance
associated with revealing the intrinsic noise introduced at that node, after
the noises indexed by $T$ have already been revealed.

Our ordered variance components have the same difference form, but are
defined directly in terms of the observed explanatory variables. By the law
of total variance,
\begin{align}
\Delta_j
&=
\E_{X_{\prec j}}
\left[
V(Y\mid X_{\prec j})
\right]
-
\E_{X_{1:j}}
\left[
V(Y\mid X_{1:j})
\right],
\label{eq:component_residual_reduction}
\end{align}
where for $j=1$ the first term is understood as $V(Y)$. Thus, $X_j$ is
attributed the reduction in expected residual variance obtained by adding it
to the variables preceding it in the chosen ordering.

The connection with ICC is especially direct for a causally sufficient model
and a fixed topological ordering. If $T$ consists of the variables preceding
$X_j$, \citet{janzing2024quantifying} show that
$\psi(Y\mid N_T)=\psi(Y\mid X_T)$.
Consequently, for the variance-based uncertainty measure, the corresponding
sequential ICC differences coincide with
\eqref{eq:component_residual_reduction}. In this special case, our
topological decomposition can therefore be viewed as an observed-data
representation of a sequential variance-based ICC decomposition.

This equivalence does not extend automatically to the more general setting
considered here. The decomposition may include only a selected collection of
observed variables from a causal system that also contains other observed or
unobserved variables. We define causal counterparts of the resulting
observed-variable components and assess their identification separately
using the full causal graph, as in
Proposition~\ref{prop:topological_identification}. An identified causal
component in this sense need not coincide with an intrinsic contribution
defined through a particular exogenous noise variable in a full structural
causal model. The two constructions therefore answer related but distinct
attribution questions: ICC attributes uncertainty to variation introduced by
node-specific mechanisms, whereas our decomposition attributes variation
sequentially to a prespecified collection of observed explanatory variables.

The treatment of ordering also differs. For any fixed ordering,
\citet{janzing2024quantifying} obtain a sequential decomposition by taking
$T$ to consist of the indices preceding each node. Their principal
construction removes dependence on this ordering by averaging the
contributions over all possible orderings using Shapley symmetrization; they
also consider an alternative in which the averaging is restricted to
topological orderings. \citet{saha2025measuring} likewise define a
topological ICC by averaging over all admissible topological orderings. In
contrast, we regard the ordering as part of the definition of the scientific
estimand. When variables are causally unordered relative to one another but
have a common scientific role, we instead propose combining them into a
vector-valued node when such grouping is scientifically appropriate. Under
the conditions of Proposition~\ref{prop:topological_clustering}, this yields
a unique topological ordering at the cluster level. The case-mix covariates
$\boldsymbol X$ in the running example are already treated in this way. This
is also consistent with the broader observation that causal attribution can
depend on the chosen granularity of the causal graph
\citep{tikka2023clustering,janzing2024quantifying}.

A further distinction concerns residual variation. In the ICC construction,
the target is a deterministic function of all exogenous noise variables, so
conditioning on all of them leaves no remaining target uncertainty. The
successive ICC differences therefore exhaust the target uncertainty when
contributions are assigned to all nodes, including the target node itself;
Shapley symmetrization preserves this property. Our decomposition instead
stops after the prespecified explanatory variables and retains the explicit
residual component
$\Delta_{\mathrm{res}} =
\E_{X_{1:k}}
\left[
V(Y\mid X_{1:k})
\right]$.
In the special causally sufficient case where $X_1,\ldots,X_k,Y$ comprise
the complete system and are taken in topological order, this residual can be
viewed as the final intrinsic contribution associated with the target
variable $Y$. More generally, however, the residual may also contain
variation attributable to variables omitted from the decomposition,
including unobserved variables, as well as intrinsic stochastic variation in
$Y$. It should therefore not in general be identified solely with the
target node's own exogenous noise.

Finally, our framework permits scientifically motivated orderings that are
not topological. As shown in
Proposition~\ref{prop:modified_order_identification}, the causal counterparts
of such components may still be identified when preceding descendants are
treated as controlled variables and the remaining preceding variables form a
valid adjustment set. These controlled-effect components do not have a direct
analogue in the topological ICC constructions considered above.

\section{Model-based plug-in point estimation}

Once a causal interpretation of a variance component has been established
using the graphical conditions of Section~3, estimation proceeds through its
identified observed-data representation. For a chosen ordering
$\sigma:\quad
X_{\sigma(1)}
\prec X_{\sigma(2)}
\prec\cdots\prec
X_{\sigma(k)}
\prec Y$,
write
$X_{\sigma(1:j)}
=
\bigl(
X_{\sigma(1)},\ldots,X_{\sigma(j)}
\bigr)$,
and define the conditional mean surface
$m(x_{1:k})
:=
\E(Y\mid X_{1:k}=x_{1:k})$.

The joint distribution can be factorized according to the chosen ordering as
\begin{equation}
p_\sigma(x_{1:k},y)
=
p_\sigma(x_{\sigma(1)})
\prod_{j=2}^{k}
p_\sigma
\left(
x_{\sigma(j)}
\mid
x_{\sigma(1:j-1)}
\right)
p(y\mid x_{1:k}).
\label{eq:order_specific_factorization}
\end{equation}
This is a statistical factorization of the observed-data distribution and
does not require $\sigma$ to be a topological ordering of the causal DAG.

To make explicit how this factorization identifies the ordered variance
decomposition, the decomposition \eqref{eq:ordered_variance_decomposition} can be written in fully nested form as
\begin{align}
V(Y)
&=
\sum_{j=1}^{k}
\E_{X_{\sigma(1)}}
\E_{X_{\sigma(2)}\mid X_{\sigma(1)}}
\cdots
\E_{X_{\sigma(j-1)}\mid X_{\sigma(1:j-2)}}
\Bigg[
\nonumber\\[-1mm]
&\qquad\quad
V_{X_{\sigma(j)}\mid X_{\sigma(1:j-1)}}
\Bigg\{ \E_{X_{\sigma(j+1)}\mid X_{\sigma(1:j)}}
\cdots
\E_{X_{\sigma(k)}\mid X_{\sigma(1:k-1)}}
\bigl[
m(X_{1:k})
\bigr]
\Bigg\}
\Bigg]
\nonumber\\
&\quad+
\E_{X_{\sigma(1)}}
\E_{X_{\sigma(2)}\mid X_{\sigma(1)}}
\cdots
\E_{X_{\sigma(k)}\mid X_{\sigma(1:k-1)}}
\left[
V(Y\mid X_{1:k})
\right].
\label{eq:ordered_variance_decomposition_nested}
\end{align}
For $j=1$, the expectations preceding the variance operator are omitted,
while for $j=k$, the expectations inside the variance operator are omitted.

Each expectation or variance in
\eqref{eq:ordered_variance_decomposition_nested} is therefore taken with
respect to one of the conditional distributions appearing in
\eqref{eq:order_specific_factorization}. By repeated application of the
tower property, the inner nested expectation for the $j$th component satisfies
\[
\E_{X_{\sigma(j+1)}\mid X_{\sigma(1:j)}}
\cdots
\E_{X_{\sigma(k)}\mid X_{\sigma(1:k-1)}}
\left[
m(X_{1:k})
\right]
=
\E
\left(
Y\mid X_{\sigma(1:j)}
\right),
\]
so that \eqref{eq:ordered_variance_decomposition_nested} is equivalent to
the more compact ordered variance decomposition introduced in Section~2.
The nested representation is useful for estimation because it makes explicit
which conditional distributions must be estimated for a given ordering.

A natural class of point estimators is consequently obtained by fitting working
models for the factors in
\eqref{eq:order_specific_factorization} and evaluating the nested
expectations and variances in
\eqref{eq:ordered_variance_decomposition_nested} by direct plug-in.
The leading factor
$p_\sigma(x_{\sigma(1)})$ may be modeled parametrically or replaced by the
empirical distribution when the first variable in the ordering is continuous
or high-dimensional. The remaining factors can be fitted using regression
models appropriate to the support of the corresponding variables.
Factorization \eqref{eq:order_specific_factorization}
aligns the fitted models directly with the
estimand. Different orderings correspond to different decomposition
functionals and, in general, require different conditional distributions.

We focus first on point estimation in the modified-order example
$\boldsymbol X \prec Z \prec A \prec Y$,
where $\boldsymbol X$ denotes a vector of case-mix covariates, $Z$ is a categorical sociodemographic group indicator, 
$A$ is categorical hospital assignment, and $Y$ is a dichotomous process-of-care outcome.
In this setting a direct order-based estimator may be based on the factorization
\[
p(\boldsymbol x,z,a,y)
=
p(\boldsymbol x)\,p(z \mid \boldsymbol x)\,p(a \mid \boldsymbol x,z)\,p(y \mid \boldsymbol x,z,a).
\]
We propose to estimate $p(\boldsymbol x)$ by the empirical distribution
$\hat F_{\boldsymbol X}$ of the observed covariates, and to fit
\[
\hat p(z \mid \boldsymbol x),\qquad
\hat p(a \mid \boldsymbol x,z),\qquad
\hat p(y=1 \mid \boldsymbol x,z,a)
\]
using multinomial logistic, multinomial logistic, and logistic regression models,
respectively.
This avoids imposing a parametric model on the joint distribution of the covariates,
while still yielding a fully model-based estimator for the reordered conditional
structure required by the decomposition.

Let $\hat m(\boldsymbol x,z,a) := \hat \E(Y \mid \boldsymbol x,z,a)$
denote the fitted conditional mean from the logistic outcome model, and define
\[
\hat g(\boldsymbol x,z)
:=
\sum_{a} \hat p(a \mid \boldsymbol x,z)\,\hat m(\boldsymbol x,z,a),
\qquad
\hat h(\boldsymbol x)
:=
\sum_{z}\hat p(z \mid \boldsymbol x)\,\hat g(\boldsymbol x,z).
\]
Then the modified-order variance components are estimated by replacing the
conditional distributions, expectations, and variances in the nested
representation \eqref{eq:ordered_variance_decomposition_nested} by their
empirical or fitted counterparts. For the ordering
$\boldsymbol X\prec Z\prec A\prec Y$, this gives
\[
\hat\Delta_{\boldsymbol X}
=
\frac{1}{n}\sum_{i=1}^n
\left\{
\hat h(\boldsymbol X_i)
-
\frac{1}{n}\sum_{i'=1}^n \hat h(\boldsymbol X_{i'})
\right\}^2,
\]
\[
\hat\Delta_Z
=
\frac{1}{n}\sum_{i=1}^n
\sum_{z}
\hat p(z \mid \boldsymbol X_i)
\left\{
\hat g(\boldsymbol X_i,z)-\hat h(\boldsymbol X_i)
\right\}^2,
\]
and
\[
\hat\Delta_A
=
\frac{1}{n}\sum_{i=1}^n
\sum_{z}\hat p(z \mid \boldsymbol X_i)
\sum_{a}\hat p(a \mid \boldsymbol X_i,z)
\left\{
\hat m(\boldsymbol X_i,z,a)-\hat g(\boldsymbol X_i,z)
\right\}^2.
\]
For a binary outcome, the residual component is estimated by
\[
\hat\Delta_{\mathrm{res}}
=
\frac{1}{n}\sum_{i=1}^n
\sum_{z}\hat p(z \mid \boldsymbol X_i)
\sum_{a}\hat p(a \mid \boldsymbol X_i,z)
\hat m(\boldsymbol X_i,z,a)\{1-\hat m(\boldsymbol X_i,z,a)\}.
\]

If one is willing to impose a parametric model for $p(\boldsymbol X \mid Z)$, 
an alternative implementation for the modified-order decomposition is to fit the
topological factorization
\[
p(z)\,p(\boldsymbol x \mid z)\,p(a \mid \boldsymbol x,z)\,p(y \mid \boldsymbol x,z,a)
\]
and recover $p(z \mid \boldsymbol x)$ by Bayes' rule. This will be of interest in the simulation study
below because the oracle data-generating mechanism uses such a model.
In realistic applications, however, a parametric model for
$p(\boldsymbol X \mid Z)$ may be less compelling than a direct model for
$Z \mid \boldsymbol X$, especially when $\boldsymbol X$ is high-dimensional or
contains a mixture of continuous and categorical variables.

\section{Approximate Bayesian inference}\label{sec:inference}

We use an approximate Bayesian procedure to quantify uncertainty in the
model-based plug-in point estimators. The construction mirrors the order-specific
factorization in \eqref{eq:order_specific_factorization}. When the
conditional factors are represented by separate working models with
variation-independent parameter blocks, their likelihood contributions
factorize. With independent flat or weakly informative priors on these
parameter blocks, this motivates drawing the parameters of the fitted
conditional models independently from large-sample Gaussian approximations
to their posterior distributions.

Let $\theta_r$ denote the parameter vector for conditional model $r$, with
maximum likelihood estimator $\hat\theta_r$ and estimated covariance matrix
$\widehat{\mathrm{Var}}(\hat\theta_r)$. For posterior draw
$b=1,\ldots,B_{\mathrm{post}}$, we independently sample
\[
\theta_r^{(b)}
\sim
N\left\{
\hat\theta_r,
\widehat{\mathrm{Var}}(\hat\theta_r)
\right\}
\]
for each fitted conditional model. For example, for the modified-order
direct estimator under $\boldsymbol X\prec Z\prec A\prec Y$,
separate draws are taken for the parameters of the models for
$Z\mid\boldsymbol X$, $A\mid\boldsymbol X,Z$, and
$Y\mid\boldsymbol X,Z,A$.

When the decomposition involves an empirical distribution, such
as the distribution of $\boldsymbol X$ in the running example, uncertainty in that
is propagated using Bayesian bootstrap, which places a nonparametric posterior distribution on the probabilities
assigned to the observed sample points. Specifically, for
posterior draw $b=1,\ldots,B_{\mathrm{post}}$, we draw
\[
(W_1^{(b)},\ldots,W_n^{(b)})
\sim
\mathrm{Dirichlet}(1,\ldots,1).
\]
Equivalently, draw independent exponential random variables
$E_i^{(b)}\sim\mathrm{Exponential}(1)$ and set
\[
W_i^{(b)}
=
\frac{E_i^{(b)}}{\sum_{j=1}^n E_j^{(b)}}.
\]
Ordinary empirical averages are then replaced by weighted averages. For example,
\[
\frac{1}{n}\sum_{i=1}^n f(\boldsymbol X_i)
\]
is replaced in posterior draw $b$ by
\[
\sum_{i=1}^n W_i^{(b)} f(\boldsymbol X_i).
\]

For posterior draw $b$, the fitted conditional probabilities and conditional
mean surface are recomputed using the sampled parameters. In the modified-order
direct estimator, define
\[
\hat g^{(b)}(\boldsymbol x,z)
:=
\sum_a
\hat p^{(b)}(a\mid \boldsymbol x,z)
\hat m^{(b)}(\boldsymbol x,z,a),
\]
where
\[
\hat m^{(b)}(\boldsymbol x,z,a)
=
\hat \E^{(b)}(Y\mid \boldsymbol x,z,a),
\]
and
\[
\hat h^{(b)}(\boldsymbol x)
:=
\sum_z
\hat p^{(b)}(z\mid \boldsymbol x)
\hat g^{(b)}(\boldsymbol x,z).
\]
The $b$th posterior draw of the modified-order direct components is then
\[
\hat\Delta_{\boldsymbol X}^{(b)}
=
\sum_{i=1}^n W_i^{(b)}
\left[
\hat h^{(b)}(\boldsymbol X_i)
-
\sum_{i'=1}^n W_{i'}^{(b)}
\hat h^{(b)}(\boldsymbol X_{i'})
\right]^2,
\]
\[
\hat\Delta_Z^{(b)}
=
\sum_{i=1}^n W_i^{(b)}
\sum_z
\hat p^{(b)}(z\mid \boldsymbol X_i)
\left\{
\hat g^{(b)}(\boldsymbol X_i,z)
-
\hat h^{(b)}(\boldsymbol X_i)
\right\}^2,
\]
\[
\hat\Delta_A^{(b)}
=
\sum_{i=1}^n W_i^{(b)}
\sum_z
\hat p^{(b)}(z\mid \boldsymbol X_i)
\sum_a
\hat p^{(b)}(a\mid \boldsymbol X_i,z)
\left\{
\hat m^{(b)}(\boldsymbol X_i,z,a)
-
\hat g^{(b)}(\boldsymbol X_i,z)
\right\}^2,
\]
and, for binary $Y$,
\[
\hat\Delta_{\mathrm{res}}^{(b)}
=
\sum_{i=1}^n W_i^{(b)}
\sum_z
\hat p^{(b)}(z\mid \boldsymbol X_i)
\sum_a
\hat p^{(b)}(a\mid \boldsymbol X_i,z)
\hat m^{(b)}(\boldsymbol X_i,z,a)
\{1-\hat m^{(b)}(\boldsymbol X_i,z,a)\}.
\]

For the topological-order ($Z \prec \boldsymbol X \prec A \prec Y$) estimator, 
the same Bayesian bootstrap draw induces both a posterior draw of the
empirical group distribution and posterior draws of the empirical
within-group distributions of $\boldsymbol X$. Thus, for posterior draw $b$,
\[
\hat p^{(b)}(z)
=
\sum_{i=1}^n W_i^{(b)} I(Z_i=z),
\]
and, within stratum $z$,
\[
W_{i\mid z}^{(b)}
=
\frac{
W_i^{(b)} I(Z_i=z)
}{
\sum_{j=1}^n W_j^{(b)} I(Z_j=z)
}.
\]
These weights replace the empirical stratum probabilities and the empirical
averages over $\boldsymbol X\mid Z=z$ in the topological decomposition. The
conditional models for $A\mid \boldsymbol X,Z$ and
$Y\mid \boldsymbol X,Z,A$ are recomputed using asymptotic normal parameter
draws, as above.

For the modified-order Bayes-inversion estimator, uncertainty in
$p(Z\mid\boldsymbol X)$ is propagated through the fitted model for
$p(\boldsymbol X\mid Z)$ together with the marginal distribution of $Z$.
For posterior draw $b$, let
$\hat p^{(b)}(Z=z)$ denote the corresponding draw of the marginal group
probability and let
$\hat f^{(b)}(\boldsymbol x\mid Z=z)$ denote the corresponding draw of the
conditional density of $\boldsymbol X$ given $Z=z$. Bayes' rule is then used
to form
\[
\hat p^{(b)}(Z=z\mid\boldsymbol X=\boldsymbol x)
=
\frac{
\hat p^{(b)}(Z=z)\,
\hat f^{(b)}(\boldsymbol x\mid Z=z)
}{
\sum_r
\hat p^{(b)}(Z=r)\,
\hat f^{(b)}(\boldsymbol x\mid Z=r)
}.
\]
The form of $\hat f^{(b)}(\boldsymbol x\mid Z=z)$ and the procedure used to
generate its uncertainty depend on the working model chosen for
$\boldsymbol X\mid Z$. The models for
$A\mid\boldsymbol X,Z$ and $Y\mid\boldsymbol X,Z,A$ are propagated using
the corresponding approximate posterior parameter draws as above.

This procedure is not intended to represent posterior inference under a
fully specified joint Bayesian model. Rather, it is a hybrid approximation
designed to propagate uncertainty through the same conditional distributions
and empirical integration steps used by the plug-in estimators. Parameters
of the fitted regression models are propagated using asymptotic Gaussian
approximations, while empirical distributions are propagated using the
Bayesian bootstrap. For the modified-order Bayes-inversion estimator, 
uncertainty in the marginal distribution of $Z$ and in the fitted model for
$\boldsymbol X\mid Z$ is propagated jointly into the Bayes inversion used
to obtain $p(Z\mid\boldsymbol X)$.
The conditional regression models are not refitted under the Bayesian bootstrap weights.

For any variance component $\Delta$, the approximate posterior sample is
$\hat\Delta^{(1)},\ldots,\hat\Delta^{(B_{\mathrm{post}})}$.
We summarize this sample by its posterior standard deviation and empirical
quantiles. The approximate 95\% posterior interval is
$
\left[
\hat q_{0.025},
\hat q_{0.975}
\right]$,
where $\hat q_\alpha$ is the empirical $\alpha$-quantile of
$\{\hat\Delta^{(b)}:b=1,\ldots,B_{\mathrm{post}}\}$.

\section{Simulation study}

We studied the finite-sample behavior of the proposed variance decomposition
estimators under a data-generating mechanism based on the running-example DAG in
Figure~\ref{fig:dag}. The objective of the simulation was to investigate the sensitivity
of the proposed point estimators to outcome model misspecification, penalization
and tuning parameter choice, as well as to produce evidence of the frequentist
properties of the proposed approximate Bayesian uncertainty quantification.

\subsection{Data-generating mechanism}
\label{sec:simulation_dgm}
The explanatory variables are sociodemographic group
$Z$, case-mix covariates $\boldsymbol X$, and hospital assignment $A$, and the
outcome $Y$ is a binary process-of-care indicator. The data-generating DAG
contains the arrows
$Z \to \boldsymbol X,$
$Z \to A$,
$Z \to Y$,
$\boldsymbol X \to A$,
$\boldsymbol X \to Y$ and
$A \to Y$.
Thus, $Z$ may affect the outcome directly and indirectly through case-mix,
hospital assignment, or both.

The variable $Z$ has three equally probable categories, denoted by
\emph{neutral}, \emph{advantaged}, and \emph{disadvantaged}. We associate with
these categories the direct-effect scores
$d_z \in \{0,1,-1\}$ and the case-mix severity scores
$s_z \in \{0,-1,1\}$, respectively. The opposite signs of the two score systems allow the advantaged
and disadvantaged groups to differ both in their clinical profiles and in
pathways operating conditionally on those profiles.

Conditional on $Z=z$, the $p$-dimensional case-mix vector is generated as
\[
\boldsymbol X \mid Z=z
\sim
N_p\left(
s_{ZX}s_z\boldsymbol\mu,\,
\boldsymbol\Sigma_X
\right),
\]
where $s_{ZX}=1.1$, $\boldsymbol\mu$ is a normalized vector with positive,
decreasing coordinates, and $\boldsymbol\Sigma_X$ is an exchangeable
correlation matrix with unit marginal variances and pairwise correlation
$\rho_X=0.25$. Consequently, the disadvantaged group has higher values of the
case-mix severity score on average, whereas the advantaged group has lower
values. We considered $p\in\{5,15\}$. In each case, only
$\lceil p/2\rceil$ coordinates of $\boldsymbol X$ had nonzero coefficients in
the hospital-assignment and outcome models. The coefficient vectors
$\boldsymbol\beta_A$ and $\boldsymbol\beta_Y$ used the same normalized,
positive, decreasing sparse pattern, whose exact construction is given in
Appendix~\ref{appendix:dgm}.

There are five hospitals, indexed by $a=1,\ldots,5$, with ordered hospital
scores $q_a\in\{-1,-0.5,0,0.5,1\}$.
Hospital assignment is generated from a multinomial logistic model,
\[
\Pr(A=a\mid \boldsymbol X,Z=z)
=
\frac{
 \exp\{\eta_a(\boldsymbol X,z)\}
}{
 \sum_{a'=1}^5 \exp\{\eta_{a'}(\boldsymbol X,z)\}
},
\]
where
\[
\eta_a(\boldsymbol X,z)
=
\alpha_a
+
s_{ZA}d_zq_a
+
s_{XA}(\boldsymbol\beta_A^\top\boldsymbol X)q_a.
\]
We set $\alpha_a=0$ for all $a$, $s_{ZA}=1$, and $s_{XA}=1.1$.
Thus, larger values of the case-mix score
$\boldsymbol\beta_A^\top\boldsymbol X$ increase the probability of assignment
to a hospital with a higher score. Conditional on case-mix, the advantaged
group is also shifted toward higher-score hospitals and the disadvantaged
group toward lower-score hospitals.

The outcome is generated according to
\[
Y\mid \boldsymbol X,Z=z,A=a
\sim
\operatorname{Bernoulli}\{m_a(\boldsymbol X,z)\},
\]
with
\begin{align}
\logit m_a(\boldsymbol X,z)
&=
\alpha_Y
+
s_{XY}\boldsymbol\beta_Y^\top\boldsymbol X
+
s_{AY}q_a
+
\lambda_z
+
\gamma_{ZA}d_zq_a.
\label{eq:simulation_outcome_model}
\end{align}
We set $\alpha_Y=-0.3$, $s_{XY}=1$, and $s_{AY}=1$. The parameter
$\gamma_{ZA}$ controls interaction between sociodemographic group and hospital
score. We considered
$\gamma_{ZA}\in\{0,0.6,1.0\}$.
When $\gamma_{ZA}>0$, the difference in outcome probability between higher- and
lower-score hospitals varies across the three groups. The data-generating
interaction is structured through the product $d_zq_a$, although the fitted
parametric interaction models described in Section~\ref{sec:simulation_estimators}
allow unrestricted interactions between the categorical variables $Z$ and
$A$.

For the baseline scenario $\gamma_{ZA}=0$, the group-specific outcome
intercepts are
$\lambda_z^{(0)}=s_{ZY}d_z$ and $s_{ZY}=0.8$.
For each $\gamma_{ZA}>0$, the intercepts $\lambda_z$ are recalibrated separately
within each group. Specifically, let
\[
\mu_z^{(0)}
=
\E_{\gamma_{ZA}=0,\lambda_z^{(0)}}(Y\mid Z=z)
\]
denote the group-specific marginal outcome mean under the baseline interaction
scenario. For each nonzero value of $\gamma_{ZA}$, $\lambda_z$ is chosen to
satisfy
\begin{align}
\E_{\gamma_{ZA},\lambda_z}(Y\mid Z=z)
&=
\mu_z^{(0)},
\qquad z=1,2,3,
\label{eq:simulation_calibration}
\end{align}
where the expectation integrates over both
$\boldsymbol X\mid Z=z$ and $A\mid\boldsymbol X,Z=z$. The calibration was
performed by Monte Carlo integration and one-dimensional root finding,
separately for each group.

This calibration preserves the group-specific marginal outcome means, and
hence the net between-group variation, as the strength of the $Z$--$A$
interaction changes. At the same time, the contributions operating through
the direct $Z\to Y$ pathway, the hospital pathway $Z\to A\to Y$, and their
interaction are altered. As a result, $\Delta_Z = V_Z\{\E(Y\mid Z)\}$
is approximately unchanged across the interaction scenarios even though the
conditional outcome response surfaces and the contributions of the remaining
components may change.

In the topological decomposition, the design and chosen parameter values create 
cancellation of the total group effect $\Delta_Z$ that combines all the distinct causal pathways. 
Under the modified ordering $\boldsymbol X\prec Z\prec A\prec Y$,
conditioning first on case-mix removes the cancellation attributable to
the pathways $Z \to \boldsymbol X \to Y$ and $Z \to \boldsymbol X \to A \to Y$, 
so that the remaining disparity due to the pathways $Z\to Y$ and $Z\to A\to Y$
can manifest.

The simulation scenarios crossed
$n\in\{500,1500\}$, $p\in\{5,15\}$ and $\gamma_{ZA}\in\{0,0.6,1.0\}$.
For each of the resulting 12 scenarios, the true topological- and
modified-order variance components were approximated by Monte Carlo
integration using a sample of size $200{,}000$ from the known
data-generating mechanism. For the modified-order decomposition, the required conditional probabilities
$\Pr(Z=z\mid\boldsymbol X)$ were obtained exactly by applying Bayes' rule to
the known $\Pr(\boldsymbol X\mid Z=z)$ distributions and the marginal
probabilities $\Pr(Z=z)$. The numerical parameter values and coefficient constructions are
summarized in Appendix~\ref{appendix:dgm}.

\subsection{Estimators being compared}\label{sec:simulation_estimators}

For each simulated data set we compute both the topological-order decomposition
for $Z \prec \boldsymbol X \prec A \prec Y$
and the modified-order decomposition for
$\boldsymbol X \prec Z \prec A \prec Y$.
We consider two classes of plug-in estimators.

The first class consists of parametric model-based estimators using 
the order-specific likelihood factorization discussed above.
For the topological ordering, we estimate the group distribution empirically,
use the empirical distribution of $\boldsymbol X$ within each group stratum, and
fit multinomial logistic and logistic regression models for
$A \mid \boldsymbol X,Z$
and $Y \mid \boldsymbol X,Z,A$,
respectively.
For the modified ordering, we estimate the empirical distribution of
$\boldsymbol X$ and fit working models for
$Z \mid \boldsymbol X$,
$A \mid \boldsymbol X,Z$,
$Y \mid \boldsymbol X,Z,A$,
using multinomial logistic, multinomial logistic, and logistic regression,
respectively.
The fitted conditional mean surface
\[
\hat m(\boldsymbol x,z,a)=\hat \E(Y \mid \boldsymbol x,z,a)
\]
and the fitted conditional probabilities are then inserted into the corresponding
variance decomposition formulas.
For the modified ordering we also consider an alternative parametric estimator
that fits a multivariate normal model for $\boldsymbol X \mid Z$ and recovers
$\Pr(Z \mid \boldsymbol X)$ by Bayes inversion.
The outcome logistic model specifications included both 
the main-effects only model $Y\sim \boldsymbol X+Z+A$
and the interaction model $Y\sim \boldsymbol X+Z*A$.
To gauge sensitivity to penalization, we also included estimators based
on Firth-corrected logistic outcome models.

The second class of estimators replaces the parametric outcome model
$Y \mid \boldsymbol X,Z,A$ by a machine-learning model while keeping the
remaining factorization models unchanged.
Specifically, we fit a gradient-boosted tree model (XGBoost) for the binary
outcome using predictors $(\boldsymbol X,Z,A)$.
Hyperparameters are tuned by cross-validated out-of-sample log-loss, reflecting
standard prediction-oriented machine-learning practice. Details of the hyperparameter
tuning are specified in Appendix \ref{appendix:tuning}.
The fitted outcome probabilities are then plugged into the same topological and
modified-order decomposition formulas as above.
This provides a less parametric plug-in estimator for the conditional mean
surface, while preserving the order-specific factorization for the remaining
distributional components. 

For each scenario, the finite-sample performance of the estimators is assessed by
comparing the estimated variance components to the scenario-specific truth
computed under the known data-generating mechanism.
We report empirical bias, standard deviation, and root mean squared error for
each variance component.
The parametric estimators serve as approximately correctly specified benchmarks,
whereas the machine-learning outcome model is included to illustrate how tuning
for out-of-sample prediction may still produce non-negligible plug-in bias in
the estimated variance decomposition components.

An uncertainty simulation study was conducted under all 12 combinations of
$n$, $p$, and $\gamma_{ZA}$ described in
Section~\ref{sec:simulation_dgm}. 
For each simulated data set, we constructed approximate Bayesian 95\%
intervals and compared them with the scenario-specific true variance components. 
Uncertainty intervals were not evaluated for the Firth-corrected or XGBoost plug-in estimators.
For each variance component, interval performance was summarized by the
ratio of empirical standard deviation of the point estimator across simulation
replicates to the mean posterior standard deviation and empirical coverage of the
95\% interval. Comparing the empirical and posterior
standard deviations assesses whether the approximate Bayesian procedure
captures the repeated-sampling variability of the plug-in estimator. Coverage
of the true component additionally reflects bias caused by outcome-model
misspecification.

All simulations and analyses were conducted in R version 4.6.1
\citep{R2026}. XGBoost models were fitted using the
\texttt{xgboost} package \citep{xgboost2026}, and Firth-corrected logistic
regression models using the \texttt{logistf} package \citep{logistf2025}.
Code reproducing the simulation study and figures is available at
\url{https://github.com/saarelao/causal-decompositions-simulations}.

\subsection{Results}

Figures~\ref{fig:simulation_error_gamma0} and
\ref{fig:simulation_error_gamma1} summarize the point-estimation results for
two endpoint interaction scenarios. When $\gamma_{ZA}=0$, the
main-effects outcome model is correctly specified and both parametric
estimators show little systematic bias for either the topological or
modified-order decomposition. Including the unnecessary $Z$-by-$A$
interaction produces broadly similar point estimates, although some small
finite-sample departures remain. Sampling variability decreases substantially
when the sample size increases from $n=500$ to $n=1500$, while the
higher-dimensional setting $p=15$ generally produces greater variability,
particularly at the smaller sample size.

The XGBoost plug-in estimator exhibits more persistent component-specific
bias. Most notably, it tends to underestimate the modified-order $\Delta_Z$ component
and the topological-order $\Delta_{\boldsymbol X}$ component, and consequently, 
tends to overestimate the residual component. These differences decrease only
partly with increasing sample size. Thus, good prediction-oriented fitting
does not by itself guarantee accurate estimation of the nonlinear variance
decomposition functionals.

\begin{figure}[!htbp]
\centering
\includegraphics[width=0.95\textwidth]{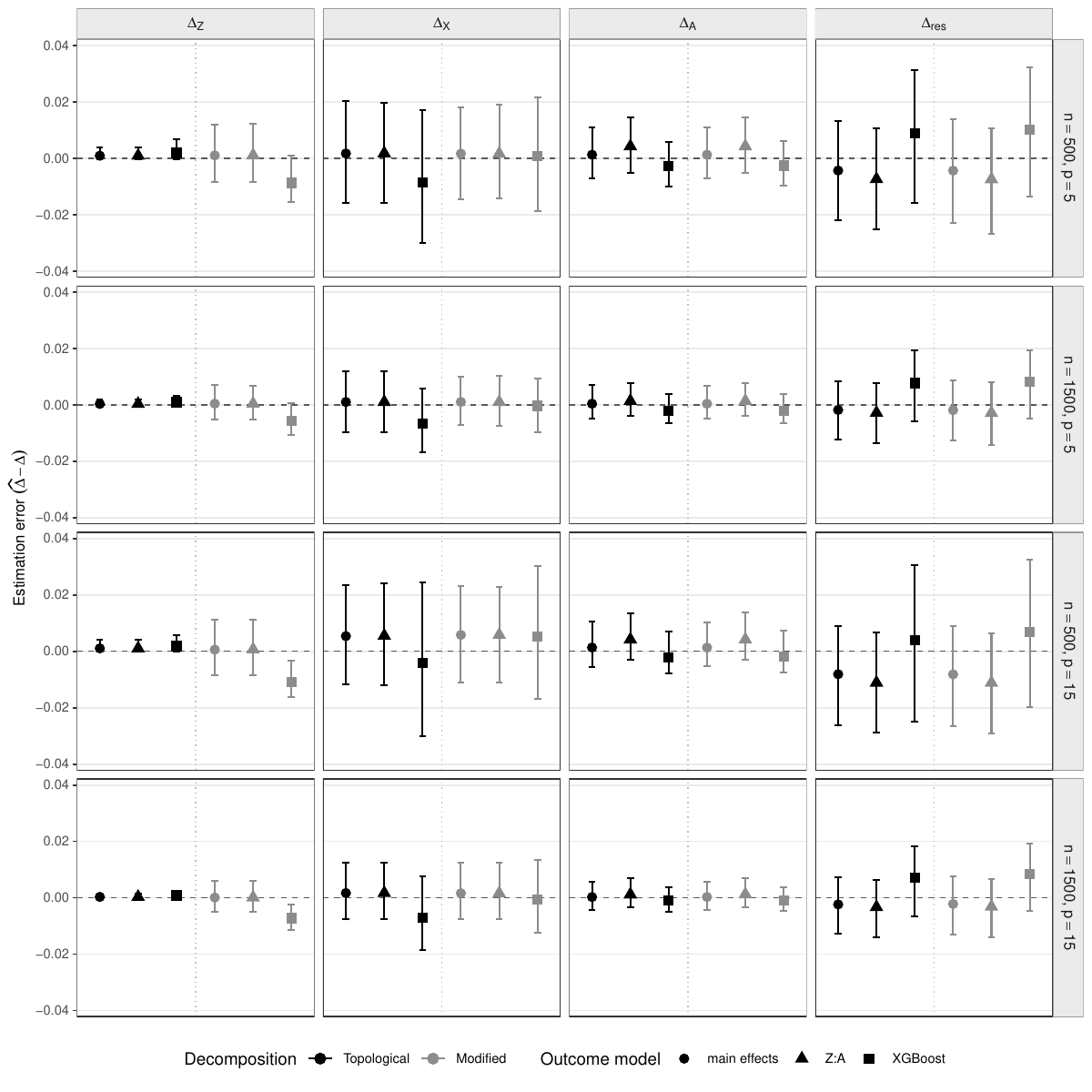}
\caption{
Monte Carlo performance of the point estimators when the data-generating mechanism outcome model has no
$Z$-by-$A$ interaction ($\gamma_{ZA}=0$).
The four columns correspond to the variance components
$\Delta_Z$, $\Delta_{\boldsymbol X}$, $\Delta_A$, and
$\Delta_{\mathrm{res}}$, and the rows correspond to the four combinations of
sample size $n$ and case-mix dimension $p$.
Points show the mean estimation error
$\hat\Delta-\Delta$ across 500 Monte Carlo replicates, with vertical bars
showing the empirical 2.5th and 97.5th percentiles of the estimation-error
distribution.
The horizontal dashed line denotes zero estimation error.
Results are shown for the topological and modified-order decompositions using
a parametric main-effects outcome model, a parametric outcome model including
the $Z$-by-$A$ interaction, and XGBoost.
}
\label{fig:simulation_error_gamma0}
\end{figure}

\begin{figure}[!htbp]
\centering
\includegraphics[width=0.95\textwidth]{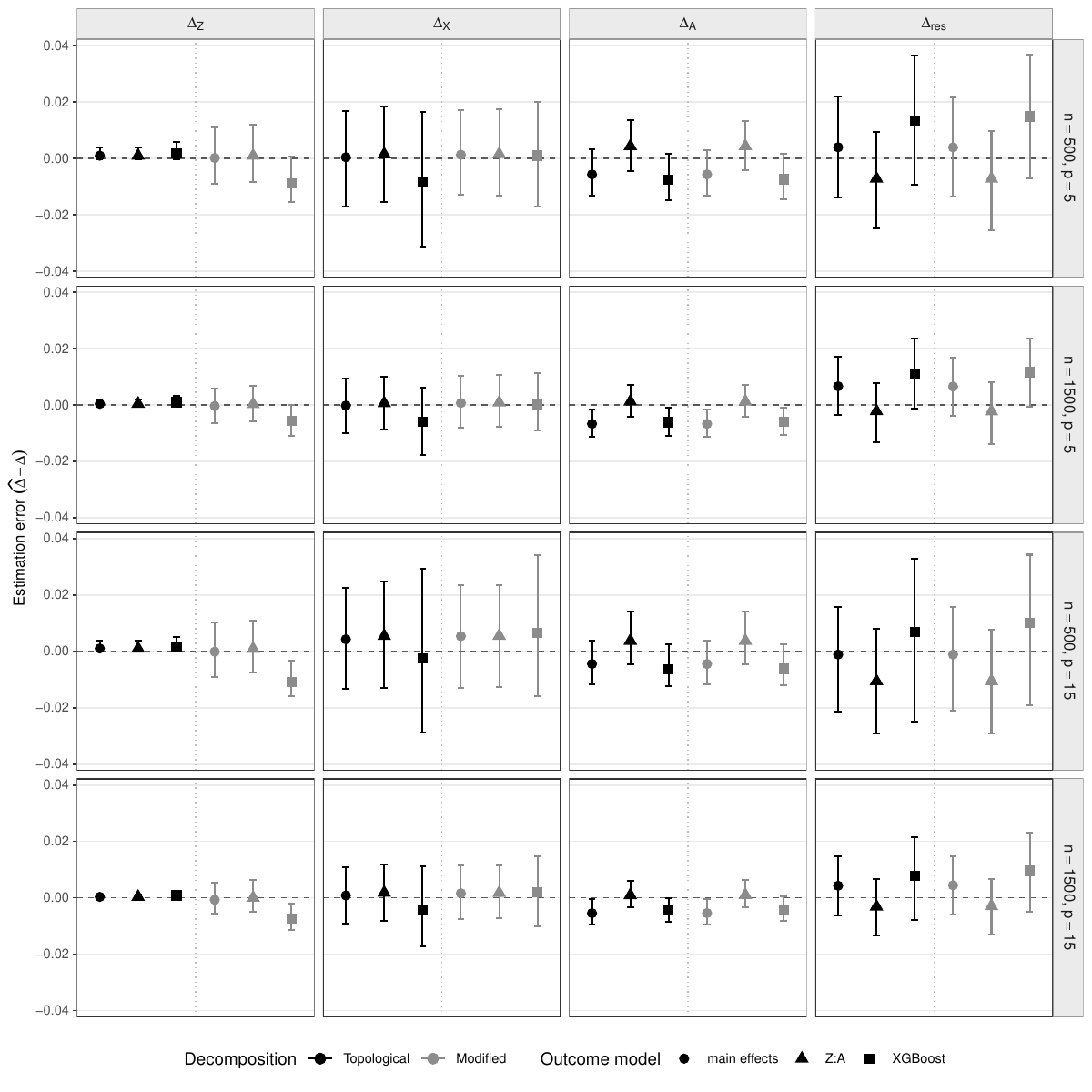}
\caption{
Monte Carlo performance of the point estimators when the data-generating mechanism outcome model contains
a $Z$-by-$A$ interaction with $\gamma_{ZA}=1$.
The four columns correspond to the variance components
$\Delta_Z$, $\Delta_{\boldsymbol X}$, $\Delta_A$, and
$\Delta_{\mathrm{res}}$, and the rows correspond to the four combinations of
sample size $n$ and case-mix dimension $p$.
Points show the mean estimation error
$\hat\Delta-\Delta$ across 500 Monte Carlo replicates, with vertical bars
showing the empirical 2.5th and 97.5th percentiles of the estimation-error
distribution.
The horizontal dashed line denotes zero estimation error.
Results are shown for the topological and modified-order decompositions.
The main-effects outcome model omits the true $Z$-by-$A$ interaction, whereas
the interaction model includes it; XGBoost provides a flexible
prediction-based alternative.
}
\label{fig:simulation_error_gamma1}
\end{figure}

When $\gamma_{ZA}=1$, the data-generating outcome model contains a substantial
$Z$-by-$A$ interaction. The fitted parametric model including this interaction
continues to estimate the decomposition components with relatively little
bias, particularly at $n=1500$. In contrast, omission of the interaction
produces systematic bias that is concentrated mainly in the hospital and
residual components: $\Delta_A$ tends to be underestimated and
$\Delta_{\mathrm{res}}$ overestimated. The group and case-mix components are
less sensitive to this particular misspecification. XGBoost again shows
non-negligible plug-in bias in several components, with patterns broadly
similar to those seen when $\gamma_{ZA}=0$.

The complete sampling distributions are shown in Supplementary
Figures~\ref{fig:supp_topological_gamma0}--\ref{fig:supp_modified_gamma1}.
These figures additionally include the intermediate interaction scenario
$\gamma_{ZA}=0.6$, the Firth-corrected interaction outcome model, and, for
the modified-order decomposition, the alternative implementation based on
Bayes inversion from the fitted distribution of $\boldsymbol X\mid Z$.
The intermediate scenario shows the expected progression between the two
endpoint cases. The direct and Bayes-inversion implementations of the
modified-order decomposition give very similar results when combined
with the same outcome model. The supplementary figures also illustrate the
substantive difference between the two decomposition targets. By design, the
topological-order $\Delta_Z$ component is essentially zero across interaction
scenarios, whereas under the modified ordering the corresponding
$\Delta_Z$ component accounts for approximately 7--8\% of the total outcome
variance, illustrating how the modified-order decomposition can capture
a path-specific signal that cancels out in the topological decomposition total effect.

Figures~\ref{fig:simulation_coverage_gamma0} and
\ref{fig:simulation_coverage_gamma1} summarize the empirical coverage of the
nominal 95\% approximate Bayesian posterior intervals. Even under correct
outcome-model specification, coverage is component-dependent. When
$\gamma_{ZA}=0$, coverage is generally close to nominal for
$\Delta_{\boldsymbol X}$ and for $\Delta_A$ under the main-effects model,
but is lower for some other components. In particular, coverage of the
topological-order $\Delta_Z$ intervals is affected by 
the true values being close to zero, in which case the asymmetric sampling
distributions are not well captured by the normal approximations. The
modified-order $\Delta_Z$ component, whose true value is appreciably away
from zero, has substantially better coverage.

\begin{figure}[!htbp]
\centering
\includegraphics[width=0.95\textwidth]{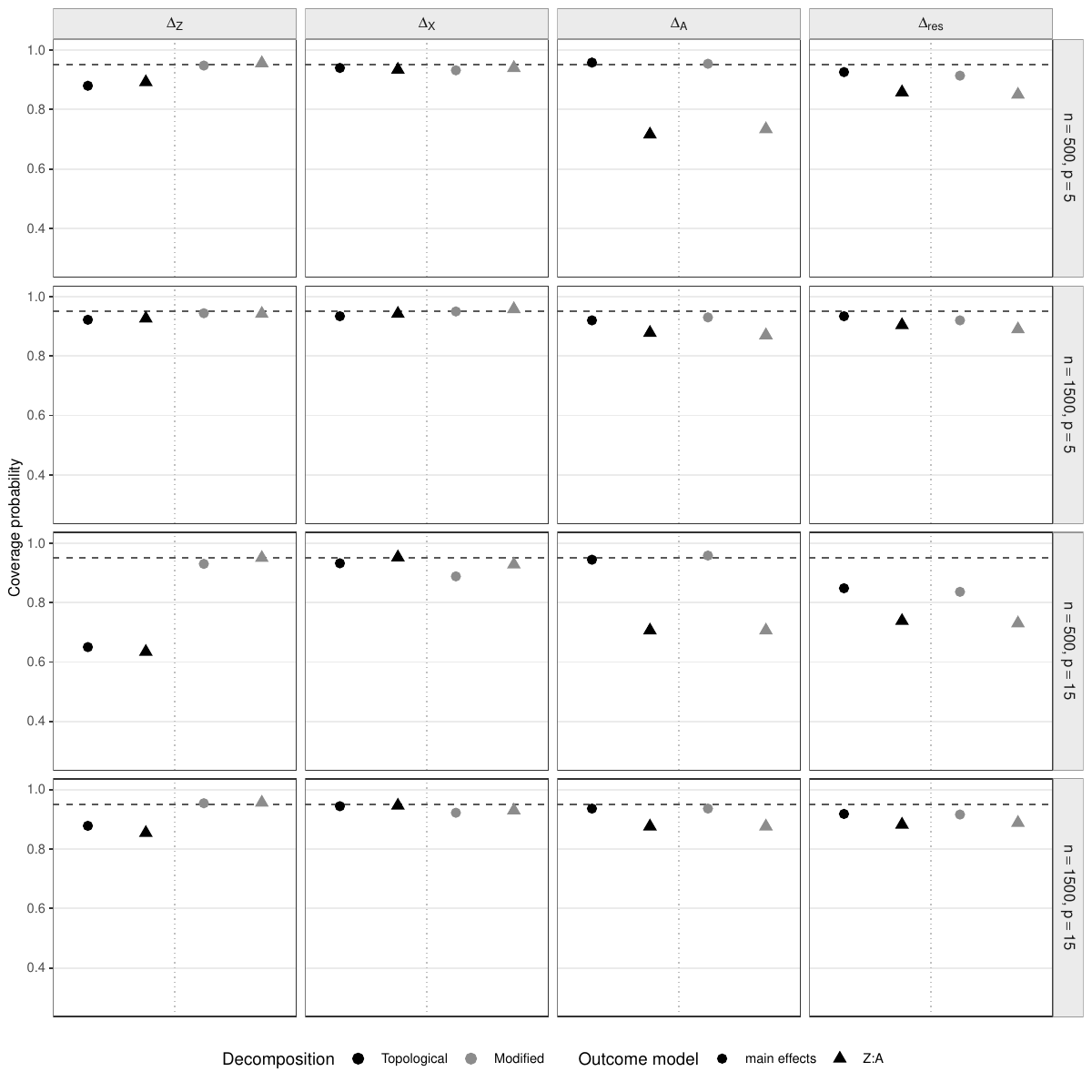}
\caption{
Empirical coverage probabilities of the nominal 95\% approximate Bayesian
posterior intervals when the data-generating mechanism outcome model has no $Z$-by-$A$ interaction
($\gamma_{ZA}=0$).
The four columns correspond to
$\Delta_Z$, $\Delta_{\boldsymbol X}$, $\Delta_A$, and
$\Delta_{\mathrm{res}}$, and the rows correspond to the four combinations of
sample size $n$ and case-mix dimension $p$.
Coverage is evaluated across 500 Monte Carlo replicates for the topological and
modified-order decompositions using the parametric main-effects and
$Z$-by-$A$ interaction outcome models.
The horizontal dashed line indicates the nominal coverage probability of
0.95.
}
\label{fig:simulation_coverage_gamma0}
\end{figure}

\begin{figure}[!htbp]
\centering
\includegraphics[width=0.95\textwidth]{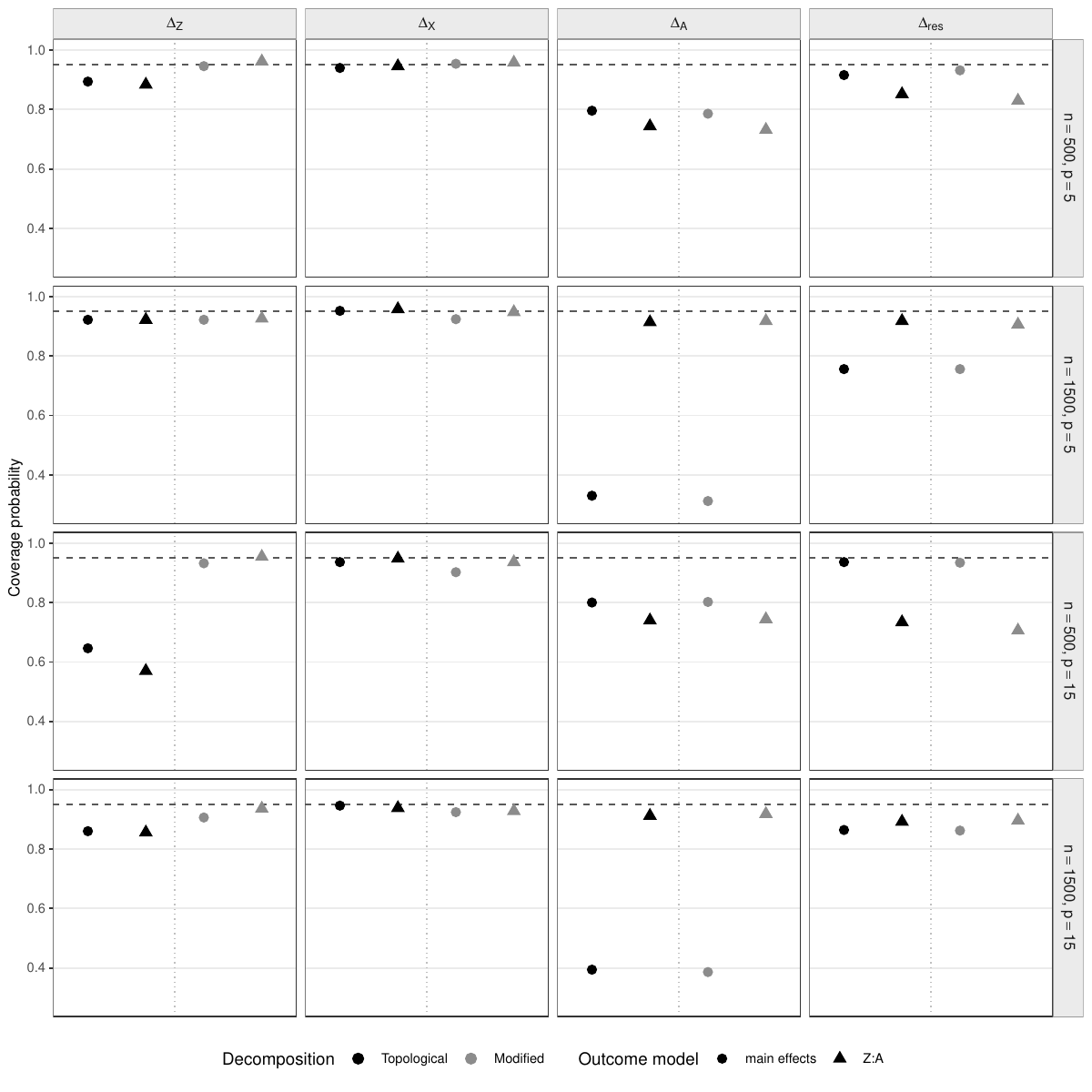}
\caption{
Empirical coverage probabilities of the nominal 95\% approximate Bayesian
posterior intervals when the data-generating mechanism outcome model contains a $Z$-by-$A$ interaction
with $\gamma_{ZA}=1$.
The four columns correspond to
$\Delta_Z$, $\Delta_{\boldsymbol X}$, $\Delta_A$, and
$\Delta_{\mathrm{res}}$, and the rows correspond to the four combinations of
sample size $n$ and case-mix dimension $p$.
Coverage is evaluated across 500 Monte Carlo replicates for the topological and
modified-order decompositions using the parametric main-effects and
$Z$-by-$A$ interaction outcome models.
The main-effects model is misspecified in this setting, whereas the interaction
model contains the data-generating outcome model.
The horizontal dashed line indicates the nominal coverage probability of
0.95.
}
\label{fig:simulation_coverage_gamma1}
\end{figure}

The unnecessary interaction model also shows some undercoverage when
$\gamma_{ZA}=0$, particularly for $\Delta_A$ and
$\Delta_{\mathrm{res}}$ at the smaller sample size. When the interaction is
present, however, the distinction between the two outcome-model
specifications becomes much more pronounced. Under $\gamma_{ZA}=1$, the
interaction model generally retains reasonable coverage for the regular
components, with performance improving as the sample size increases. In
contrast, the misspecified main-effects model results in undercoverage,
especially for $\Delta_A$, explained by the bias in the point estimates. 
The undercoverage can become more pronounced as
$n$ increases because the systematic plug-in bias persists while the
uncertainty intervals become narrower.

Supplementary
Figures~\ref{fig:supp_topological_uncertainty_gamma0}--%
\ref{fig:supp_modified_uncertainty_gamma1}
provide the full uncertainty diagnostics for all three interaction scenarios.
In addition to coverage, they compare the empirical standard deviation of
the point estimator with the mean posterior standard deviation. For many
components these quantities are reasonably similar, particularly under the
correctly specified interaction model and for $\Delta_{\boldsymbol X}$.
The supplementary results also show broadly similar uncertainty behavior for the direct and
Bayes-inversion implementations of the modified-order decomposition.
Overall, the simulations indicate that the proposed plug-in and approximate
Bayesian procedures perform reasonably under appropriate outcome-model
specification, but that both point estimation and interval coverage can be
sensitive to outcome-model misspecification and to components lying near the
boundary at zero.

\FloatBarrier

\section{Discussion}

We developed a general graph-based framework for assigning causal
interpretations to components of ordered variance decompositions. The main
distinction from existing causal variance decomposition and attribution
approaches is that identification is assessed component by component in the
full causal graph, without requiring causal sufficiency with respect to the variables included in the decomposition, 
and that the conditioning order is treated as part of the scientific estimand rather than
as an arbitrariness that must necessarily be averaged away. The framework also
allows scientifically motivated non-topological orderings, with preceding
descendants interpreted as controlled variables, and permits vector-valued
clustering when a coarser attribution is more meaningful. In our applied context,
our controlled-effect construction can help reveal between-group disparities 
that remain after blocking specific causal pathways.

Variance decompositions are attractive for causal attribution because they
express contributions on a common additive scale. After normalization by the
total outcome variance, each component can be interpreted as the proportion
of overall variation attributable to a particular source, yielding a
dimensionless measure that is invariant to linear rescaling of the outcome.
The construction also accommodates categorical explanatory variables
directly, which is particularly useful here for polytomous group and hospital
variables, and extends naturally to vector-valued nodes, allowing related
variables to be attributed jointly as a cluster.

While our simulation results support validity of the proposed procedures
under correct model specification, they reveal sensitivity to misspecification and hyperparameter tuning.
Thus, the simulation study highlights estimation as an important area for further development.
The proposed estimators are plug-in estimators of nonlinear
functionals involving several fitted conditional distributions, and the
results show that bias in these nuisance-model estimates can
translate into bias in the resulting variance components. This
was apparent in particular for the prediction-oriented XGBoost estimator,
illustrating that good predictive performance does not necessarily imply
accurate estimation of the decomposition functional. A natural next step is
therefore to derive efficient influence functions for the individual
components and use them to construct bias-corrected one-step or targeted
estimators. Such methods could also permit flexible nuisance estimation with
cross-fitting while reducing first-order sensitivity of the variance
components to nuisance-model estimation error. Related work on do-Shapley values 
has used doubly robust, debiased machine-learning estimators based on orthogonal-style 
representations of the underlying interventional means \citep{jung2022measuring}.

Bayesian uncertainty quantification for the variance components can also be developed further. 
The approximate Bayesian procedure considered here combines large-sample Gaussian
approximations for regression parameters with Bayesian-bootstrap perturbations
of empirical distribution factors. Although computationally convenient, the
simulation results show that the resulting nominal 95\% intervals can have
undercoverage in particular if the sampling distribution is concentrated near zero. 
A proper Bayesian implementation based on posterior sampling by MCMC could avoid reliance on
Gaussian approximations and propagate uncertainty jointly through the fitted
conditional models. Such an approach may improve finite-sample posterior
approximation, although it would not by itself remove bias arising from
misspecified outcome or other nuisance models. Combining richer Bayesian
models with appropriate regularization therefore represents another useful
direction for investigation.

The present analysis treats the causal graph and the scientifically relevant ordering or clustering as given.
In applications, these choices may themselves be uncertain, and sensitivity
analysis over plausible graphs, orderings, or cluster definitions may be
important. Likewise, although clustering can resolve some ambiguities caused
by parallel causal structures, the appropriate level of granularity is
ultimately substantive and can change the attribution estimand.
Overall, the results suggest that graphical identification of causal variance
components is only one part of the problem. Once the causal estimand has been
defined by the graph and ordering, robust estimation of the resulting
observed-data functional becomes equally important. Developing estimators and
inferential procedures that retain the interpretability of the present
framework while being less sensitive to nuisance-model specification is
therefore a central direction for future work.

An additional qualification concerns sociodemographic group variables such
as race or ethnicity. Because such attributes are not themselves
well-defined intervention targets, causal disparity analyses often avoid
interpreting contrasts indexed by group membership as effects of literally
intervening on group status \citep{naimi2016mediation,jackson2021meaningful}.
We do not require such a literal intervention interpretation here. Rather,
the causal framework is used to distinguish pathways from group membership
to the outcome and to formalize which descendant covariates are controlled
or regarded as allowable for adjustment. The resulting group component is
therefore interpreted as a pathway-structured disparity contrast, rather
than as the effect of an intervention that changes group membership.
To quantify disparities that remain after hypothetical interventions on
downstream variables that are well-defined intervention targets, the present
framework can be extended to interventional causal variance decompositions,
in which the variance of an interventional potential outcome is decomposed.
For example, hospital assignment may be drawn from a specified intervention
distribution, as in \citet{chen2020causal}.

\subsection*{Acknowledgement}

The scientific questions, conceptual contributions, and substantive
interpretations in this work originated with the authors, who take full
responsibility for the content of the manuscript. In preparing this work,
the authors experimented with an AI-assisted research and writing workflow. 
Chatbots (Microsoft 365 Copilot and ChatGPT, using GPT-5.5 and
GPT-5.6 models) were used to assist with drafting manuscript sections and
developing mathematical notation and formulations based on ideas and
instructions provided by the authors. All resulting text and mathematical
content were reviewed, checked, and edited by the authors. A coding agent
(OpenAI Codex, using GPT-5.5 and GPT-5.6 models) was used to implement the
simulation studies and assist in summarizing their results.

The work of OS was supported by Discovery Grants (RGPIN-2020-05920 and RGPIN-2026-06817) 
from the Natural Sciences and Engineering Research Council of Canada. The work of JK was supported by the Research Council of Finland (grant no. 368935).

\bibliographystyle{apalike}
\bibliography{refs_updated}

\appendix

\setcounter{figure}{0}
\renewcommand{\thefigure}{A\arabic{figure}}
\renewcommand{\figurename}{Supplementary Figure}

\section{Supplementary material: simulation study details and full results}

\subsection{Additional details of the data-generating mechanism}
\label{appendix:dgm}

This appendix provides the parameter constructions and numerical procedures
needed to reproduce the data-generating mechanism described in
Section~\ref{sec:simulation_dgm}.

\paragraph{Case-mix parameters.}

The unnormalized mean-direction vector
$\widetilde{\boldsymbol\mu}
=(\widetilde\mu_1,\ldots,\widetilde\mu_p)^\top$
has entries
\[
\widetilde\mu_j
=
1-\frac{0.4(j-1)}{p-1},
\qquad j=1,\ldots,p.
\]
Thus, its entries are equally spaced from $1$ to $0.6$. The normalized vector
used in the conditional mean of $\boldsymbol X$ is
\[
\boldsymbol\mu
=
\frac{\widetilde{\boldsymbol\mu}}
{\|\widetilde{\boldsymbol\mu}\|_2}.
\]
The common covariance matrix of $\boldsymbol X\mid Z$ is
\[
\boldsymbol\Sigma_X
=
(1-\rho_X)\boldsymbol I_p
+
\rho_X\boldsymbol 1_p\boldsymbol 1_p^\top,
\qquad
\rho_X=0.25.
\]

Let
\[
r=\left\lceil\frac{p}{2}\right\rceil
\]
denote the number of active case-mix variables. The unnormalized sparse
coefficient vector
$\widetilde{\boldsymbol\beta}
=(\widetilde\beta_1,\ldots,\widetilde\beta_p)^\top$
has entries
\[
\widetilde\beta_j
=
\begin{cases}
\displaystyle
1-\frac{0.5(j-1)}{r-1},
& j=1,\ldots,r,\\[8pt]
0,
& j=r+1,\ldots,p.
\end{cases}
\]
Thus, the first $r$ entries are equally spaced from $1$ to $0.5$, and the
remaining $p-r$ entries are zero. The normalized coefficient vector is
\[
\boldsymbol\beta
=
\frac{\widetilde{\boldsymbol\beta}}
{\|\widetilde{\boldsymbol\beta}\|_2},
\]
and the same coefficient pattern is used in the hospital-assignment and
outcome models:
$\boldsymbol\beta_A
=
\boldsymbol\beta_Y
=
\boldsymbol\beta$.

\paragraph{Fixed numerical parameters.}
The parameters held fixed across simulation scenarios are
$s_{ZX}=1.1$,
$s_{ZA}=1$,
$s_{XA}=1.1$,
$s_{AY}=1$,
$s_{ZY}=0.8$,
$s_{XY}=1$ and
$\alpha_Y=-0.3$.
All hospital-assignment intercepts are zero,
$\alpha_a=0$, $a=1,\ldots,5$,
and the hospital scores are
$(q_1,\ldots,q_5)=(-1,-0.5,0,0.5,1)$.

\paragraph{Calibration of the group-specific outcome intercepts.}
For $\gamma_{ZA}>0$, the group-specific intercept $\lambda_z$ is obtained by
solving
\[
F_z(\lambda;\gamma_{ZA})=\mu_z^{(0)},
\]
where
\begin{align*}
F_z(\lambda;\gamma_{ZA})
&=
\int_{\mathbb R^p}
\sum_{a=1}^5
\Pr(A=a\mid \boldsymbol x,Z=z)
\\
&\quad\times
\textrm{expit}\left\{
\alpha_Y
+s_{XY}\boldsymbol\beta_Y^\top\boldsymbol x
+s_{AY}q_a
+\lambda
+\gamma_{ZA}d_zq_a
\right\}
\\
&\quad\times
\phi_p\!\left(
\boldsymbol x;
s_{ZX}s_z\boldsymbol\mu,
\boldsymbol\Sigma_X
\right)
\,d \boldsymbol x.
\end{align*}
and
\[
\mu_z^{(0)}
=
F_z\left(s_{ZY}d_z;0\right).
\]
The integral was approximated using $200{,}000$ draws from
$\boldsymbol X\mid Z=z$. The same draws were retained for all candidate values
of $\lambda$ within a group, and the resulting one-dimensional equation was
solved numerically. Calibration was regarded as successful when the absolute
difference between the achieved and target group-specific marginal means was
less than $0.002$.

\paragraph{Monte Carlo population values.}
The population topological- and modified-order variance components were
approximated using Monte Carlo samples of size $200{,}000$. For the modified
ordering, the conditional group probabilities were evaluated from the known
data-generating mechanism as
\[
\Pr(Z=z\mid\boldsymbol X=\boldsymbol x)
=
\frac{
 \Pr(Z=z)\,
 \phi_p\!\left(
 \boldsymbol x;
 s_{ZX}s_z\boldsymbol\mu,
 \boldsymbol\Sigma_X
 \right)
}{
 \sum_{z'}
 \Pr(Z=z')\,
 \phi_p\!\left(
 \boldsymbol x;
 s_{ZX}s_{z'}\boldsymbol\mu,
 \boldsymbol\Sigma_X
 \right)
},
\]
where
$\phi_p(\,\cdot\,;\boldsymbol m,\boldsymbol\Sigma)$ denotes the
$p$-variate normal density with mean $\boldsymbol m$ and covariance matrix
$\boldsymbol\Sigma$.

As an implementation check, when $\gamma_{ZA}=0$, the population components
from the interaction-capable data-generating mechanism were required to agree
with those from the corresponding implementation without the interaction,
within a Monte Carlo tolerance of $0.003$.

\subsection{XGBoost tuning specification}
\label{appendix:tuning}

For the machine-learning plug-in estimators, the binary outcome regression
\[
Y \mid \boldsymbol X,Z,A
\]
is replaced by a gradient-boosted tree model fit using XGBoost with the binary
logistic objective. The remaining factors in the order-specific likelihood
factorizations are estimated as for the corresponding parametric estimators.
In particular, the modified-order estimator uses multinomial logistic models
for $Z\mid\boldsymbol X$ and $A\mid\boldsymbol X,Z$, whereas the
topological-order estimator uses the empirical distribution of $Z$ and the
empirical distribution of $\boldsymbol X$ within levels of $Z$.

The predictors supplied to the XGBoost outcome model are
$(\boldsymbol X,Z,A)$. No explicit $Z$--$A$ product terms are included in the
design matrix, but interactions between $Z$ and $A$, as well as interactions
involving $\boldsymbol X$, can be represented through successive tree splits.

Hyperparameters are selected using five-fold cross-validation with
out-of-sample binary log-loss as the evaluation criterion. The tuning
grid was chosen as
$\eta
\in
\{0.03,\;0.05,\;0.1,\;0.2\}$,
$\texttt{max\_depth}
\in
\{2,\;3,\;4,\;5,\;6\}$, and
$\texttt{min\_child\_weight}
\in
\{0.5,\;1,\;2,\;3\}$.
The tuning procedure therefore compares
$4\times 5\times 4=80$
hyperparameter combinations.

The remaining XGBoost hyperparameters are held fixed at
$\texttt{subsample}=0.8$,
$\texttt{colsample\_bytree}=0.8$,
$\lambda=1$,
$\alpha=0$ and
$\gamma=0$.
For each hyperparameter combination, training is allowed to continue for at
most $2{,}000$ boosting rounds, with early stopping after 50 rounds without
improvement in the cross-validated validation-set log-loss. The number of
boosting rounds associated with the lowest validation loss is retained for
that combination.

The selected hyperparameter combination is the one attaining the smallest
cross-validated log-loss at its retained number of boosting rounds. The final
model is then fit to the full simulated data set using the selected
hyperparameters and number of boosting rounds. Its predictions provide
\[
\hat m(\boldsymbol x,z,a)
=
\widehat{\Pr}(Y=1\mid
\boldsymbol X=\boldsymbol x,Z=z,A=a),
\]
which is inserted into the topological- and modified-order variance
decomposition formulas.

This tuning strategy reflects standard prediction-oriented
machine-learning practice: model complexity is selected to optimize
out-of-sample predictive performance rather than the variance decomposition
estimand. The XGBoost simulations therefore assess whether a flexible outcome
model tuned for predictive accuracy also performs well when its predictions
are used inside the nonlinear variance decomposition functionals.

\subsection{Uncertainty simulation specification}
\label{appendix:uncertainty}

The uncertainty simulation was conducted under all 12 data-generating
scenarios defined in Section~\ref{sec:simulation_dgm}. We considered two
ordinary logistic-regression specifications for the conditional outcome mean:
\[
\logit\Pr(Y=1\mid\boldsymbol X,Z,A)
=
\theta_0
+\boldsymbol\theta_X^\top\boldsymbol X
+\boldsymbol\theta_Z^\top\boldsymbol d_Z(Z)
+\boldsymbol\theta_A^\top\boldsymbol d_A(A),
\]
and
\[
\logit\Pr(Y=1\mid\boldsymbol X,Z,A)
=
\theta_0
+\boldsymbol\theta_X^\top\boldsymbol X
+\boldsymbol\theta_Z^\top\boldsymbol d_Z(Z)
+\boldsymbol\theta_A^\top\boldsymbol d_A(A)
+\boldsymbol\theta_{ZA}^\top
\left\{
\boldsymbol d_Z(Z)\otimes\boldsymbol d_A(A)
\right\},
\]
where $\boldsymbol d_Z(Z)$ and $\boldsymbol d_A(A)$ denote dummy-variable
representations of the categorical variables $Z$ and $A$, respectively.
The second specification includes an unrestricted $Z$-by-$A$ interaction
through all products of the corresponding dummy variables. The main-effects model is correctly
specified when $\gamma_{ZA}=0$ and misspecified when $\gamma_{ZA}>0$, whereas
the interaction model contains the structured data-generating interaction in
all scenarios.

Each outcome-model specification was combined with the following three
decomposition implementations:
\begin{enumerate}
\item the topological-order estimator for
$Z\prec\boldsymbol X\prec A\prec Y$;
\item the modified-order estimator for
$\boldsymbol X\prec Z\prec A\prec Y$,
using a directly fitted multinomial logistic model for
$Z\mid\boldsymbol X$; and
\item the same modified-order estimator with
$\Pr(Z\mid\boldsymbol X)$ obtained by Bayes inversion from a fitted
multivariate normal model for $\boldsymbol X\mid Z$.
\end{enumerate}
Thus, six ordinary-logistic estimator configurations were evaluated in each
data-generating scenario. Uncertainty intervals were not computed for the
Firth-corrected logistic or XGBoost plug-in estimators.

For each simulated data set and estimator configuration, the relevant
conditional models were fit once. Approximate posterior draws of the
regression parameters were generated from
\[
N\left(
\widehat{\boldsymbol\theta},
\widehat{\operatorname{Var}}
(\widehat{\boldsymbol\theta})
\right),
\]
using the estimated covariance matrix from the corresponding ordinary maximum
likelihood fit. Under the order-specific likelihood factorization, draws for
the different fitted conditional models were generated independently.

Bayesian bootstrap weights were used to propagate uncertainty in empirical
integration distributions. The conditional regression models were not refit
under these weights. For the topological-order estimator, the weights were used
in estimating the empirical distribution of $Z$ and the empirical
distributions of $\boldsymbol X$ within levels of $Z$. For the modified-order
estimators, the weights were used in integration over the empirical marginal
distribution of $\boldsymbol X$. For the Bayes-inversion estimator, they were
also used in estimating $\Pr(Z=z)$, the group-specific means of
$\boldsymbol X\mid Z=z$, and the common covariance matrix of
$\boldsymbol X\mid Z$.

For each simulation replicate $r=1,\ldots,B$, predictions from all required
conditional models were recalculated for each posterior draw and inserted
into the appropriate variance decomposition functional. This produced
posterior draws
\[
\hat\Delta_r^{(1)},\ldots,
\hat\Delta_r^{(B_{\mathrm{post}})}
\]
for every variance component. The approximate 95\% interval for replicate
$r$ was defined by the empirical posterior quantiles
\[
\left[
\hat q_{0.025,r},
\hat q_{0.975,r}
\right].
\]

Across the $B$ repeated simulation samples within each scenario, interval
performance for a component with true value $\Delta_0$ was summarized by
\[
\mathrm{EmpSD}
=
\operatorname{sd}
\left\{
\hat\Delta_1,\ldots,\hat\Delta_B
\right\},
\]
\[
\overline{\mathrm{PostSD}}
=
\frac{1}{B}
\sum_{r=1}^{B}
\widehat{\operatorname{sd}}_{\mathrm{post},r},
\]
\[
\mathrm{Coverage}
=
\frac{1}{B}
\sum_{r=1}^{B}
I\left\{
\Delta_0
\in
\left[
\hat q_{0.025,r},
\hat q_{0.975,r}
\right]
\right\},
\]
and
\[
\overline{\mathrm{Length}}
=
\frac{1}{B}
\sum_{r=1}^{B}
\left(
\hat q_{0.975,r}
-
\hat q_{0.025,r}
\right).
\]
Here $\hat\Delta_r$ is the point estimate from simulation replicate $r$,
$\widehat{\operatorname{sd}}_{\mathrm{post},r}$ is the standard deviation of
the corresponding posterior draws
$\hat\Delta_r^{(1)},\ldots,\hat\Delta_r^{(B_{\mathrm{post}})}$, and
$\hat q_{0.025,r}$ and $\hat q_{0.975,r}$ are the corresponding posterior
interval endpoints.

For the misspecified main-effects outcome model when $\gamma_{ZA}>0$,
empirical coverage reflects both the estimated sampling uncertainty and the
systematic plug-in bias induced by omission of the $Z$--$A$ interaction.

\FloatBarrier

\subsection{Supplementary figures}
\label{appendix:suptables}

\begin{figure}[!htbp]
\centering
\includegraphics[width=1.0\textwidth]
{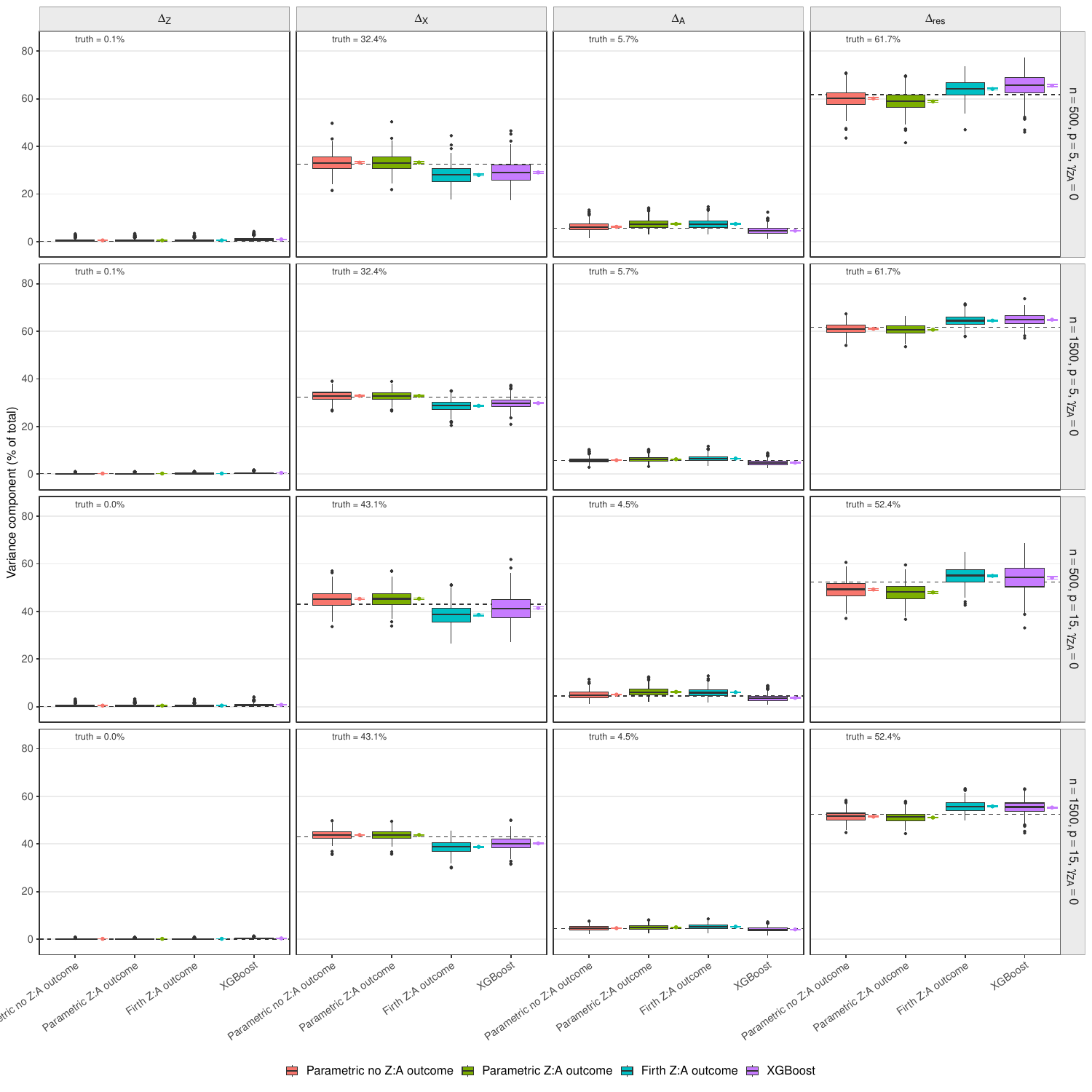}
\caption{
Finite-sample distributions of the estimated topological-order variance
components when $\gamma_{ZA}=0$, expressed as percentages of the total
outcome variance.
The four columns correspond to
$\Delta_Z$, $\Delta_{\boldsymbol X}$, $\Delta_A$, and
$\Delta_{\mathrm{res}}$, and the rows correspond to the four combinations of
sample size $n$ and case-mix dimension $p$.
Boxplots summarize estimates across 500 Monte Carlo replicates.
The horizontal dashed line in each panel denotes the corresponding true
variance component.
The small colored point and interval show the Monte Carlo mean estimate and
its 95\% normal-approximation Monte Carlo error interval.
Results are shown for the parametric main-effects outcome model, the
parametric outcome model including the $Z$-by-$A$ interaction, Firth
logistic regression including the interaction, and XGBoost.
}
\label{fig:supp_topological_gamma0}
\end{figure}

\begin{figure}[!htbp]
\centering
\includegraphics[width=1.0\textwidth]
{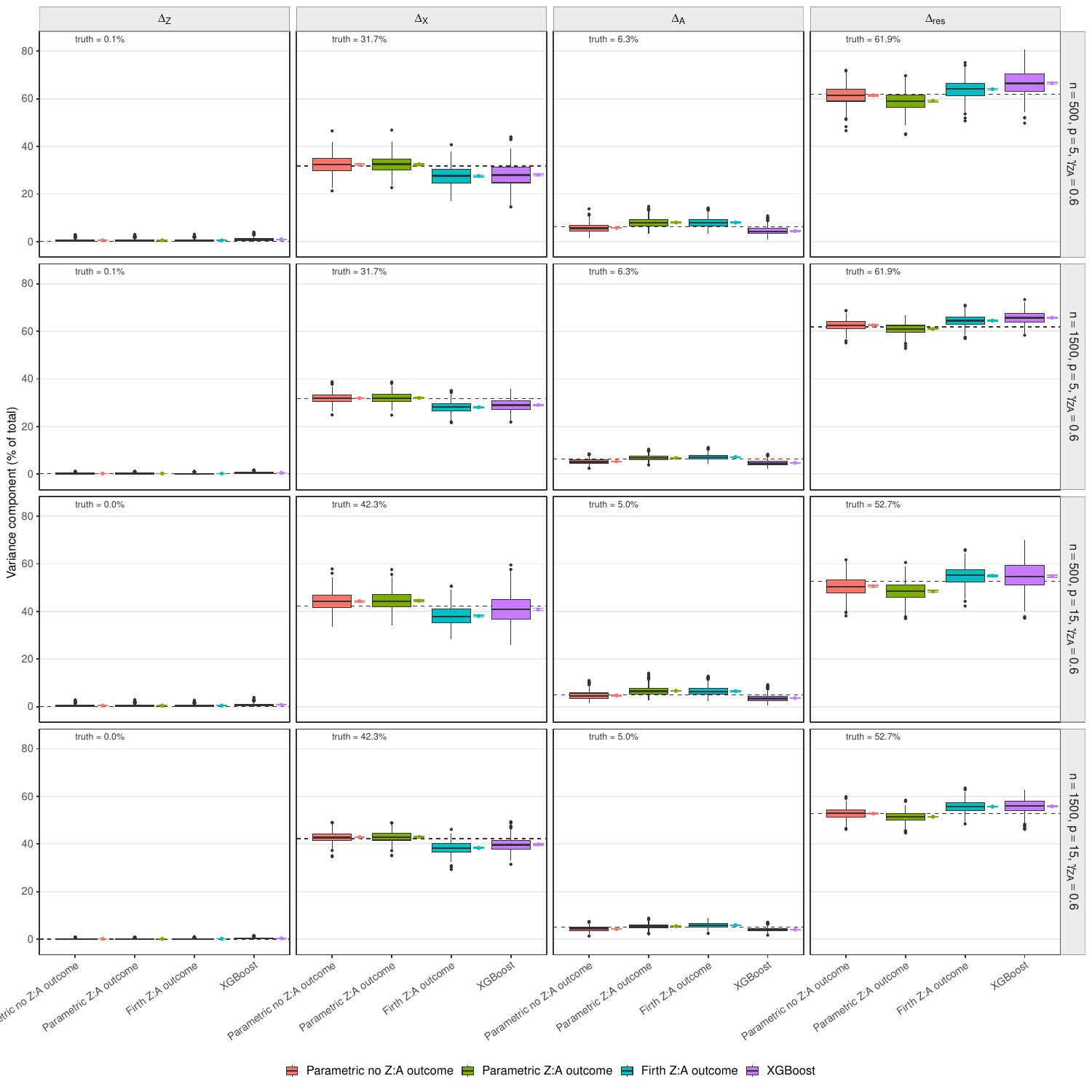}
\caption{
Finite-sample distributions of the estimated topological-order variance
components when $\gamma_{ZA}=0.6$, expressed as percentages of the total
outcome variance.
The four columns correspond to
$\Delta_Z$, $\Delta_{\boldsymbol X}$, $\Delta_A$, and
$\Delta_{\mathrm{res}}$, and the rows correspond to the four combinations of
sample size $n$ and case-mix dimension $p$.
Boxplots summarize estimates across 500 Monte Carlo replicates.
The horizontal dashed line in each panel denotes the corresponding true
variance component.
The small colored point and interval show the Monte Carlo mean estimate and
its 95\% normal-approximation Monte Carlo error interval.
Results are shown for the parametric main-effects outcome model, the
parametric outcome model including the $Z$-by-$A$ interaction, Firth
logistic regression including the interaction, and XGBoost.
}
\label{fig:supp_topological_gamma06}
\end{figure}

\begin{figure}[!htbp]
\centering
\includegraphics[width=1.0\textwidth]
{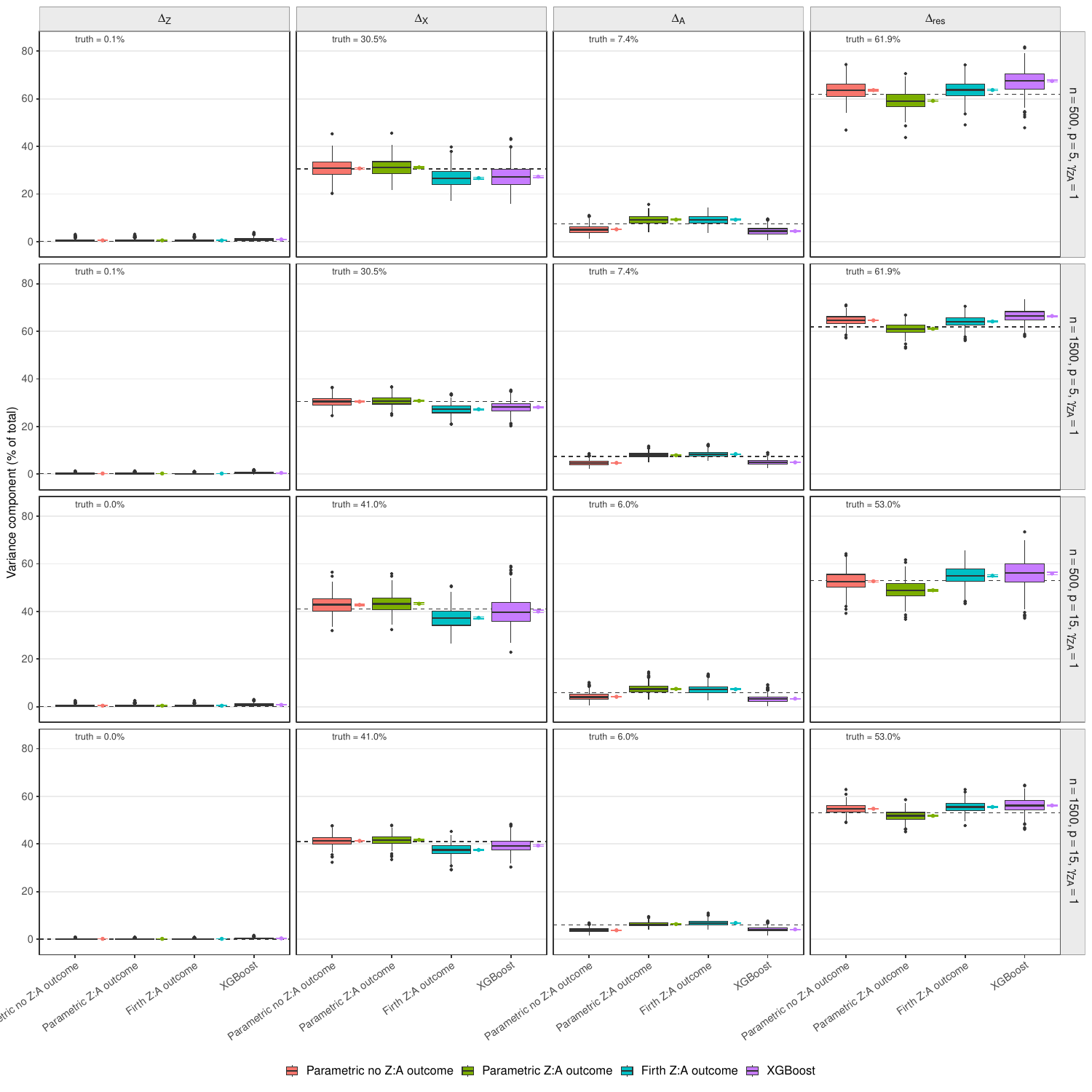}
\caption{
Finite-sample distributions of the estimated topological-order variance
components when $\gamma_{ZA}=1$, expressed as percentages of the total
outcome variance.
The four columns correspond to
$\Delta_Z$, $\Delta_{\boldsymbol X}$, $\Delta_A$, and
$\Delta_{\mathrm{res}}$, and the rows correspond to the four combinations of
sample size $n$ and case-mix dimension $p$.
Boxplots summarize estimates across 500 Monte Carlo replicates.
The horizontal dashed line in each panel denotes the corresponding true
variance component.
The small colored point and interval show the Monte Carlo mean estimate and
its 95\% normal-approximation Monte Carlo error interval.
The parametric main-effects outcome model omits the true $Z$-by-$A$
interaction, whereas the interaction and Firth models include it; XGBoost
provides a flexible prediction-based alternative.
}
\label{fig:supp_topological_gamma1}
\end{figure}

\begin{figure}[!htbp]
\centering
\includegraphics[width=1.0\textwidth]
{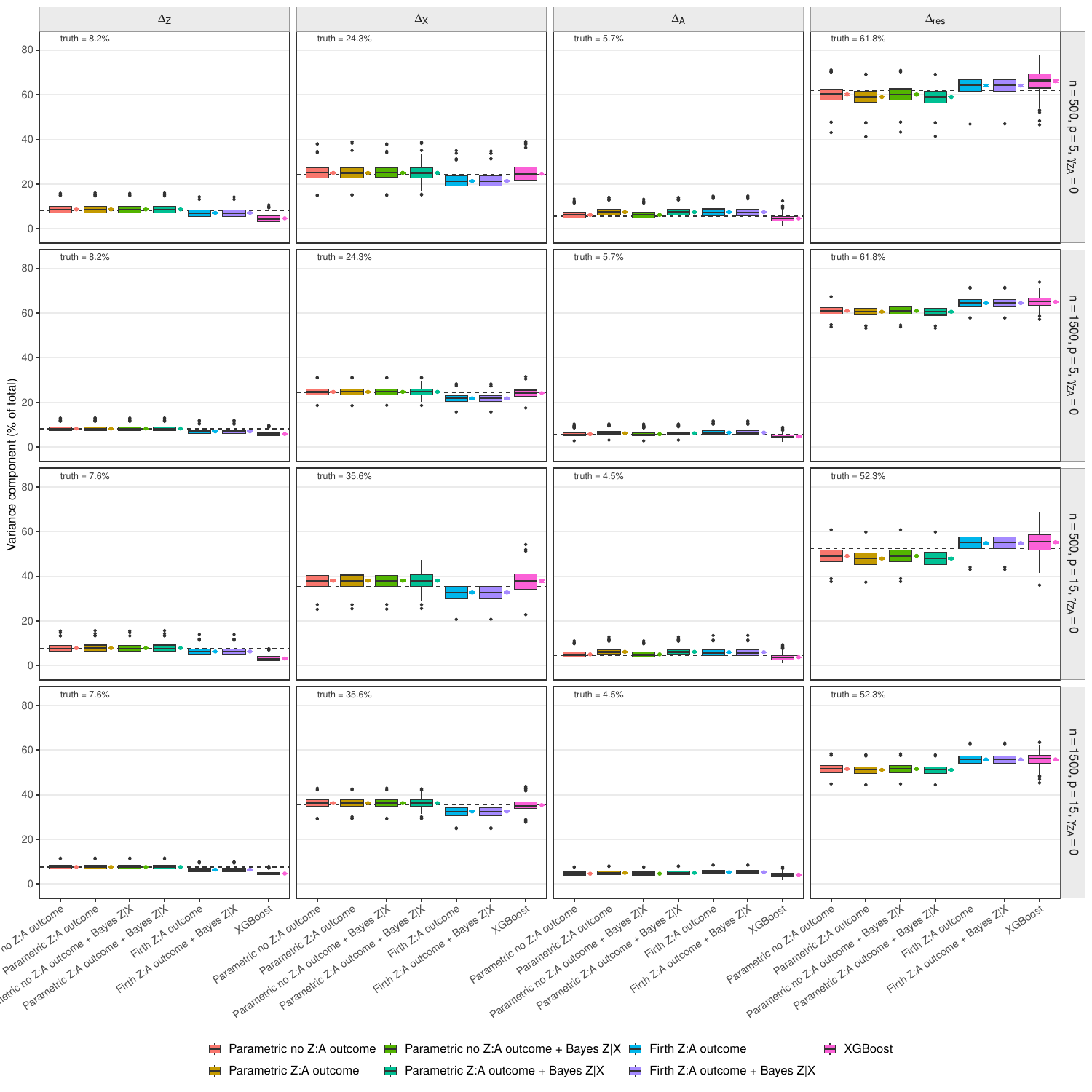}
\caption{
Finite-sample distributions of the estimated modified-order variance
components when $\gamma_{ZA}=0$, expressed as percentages of the total
outcome variance.
The modified ordering is
$\boldsymbol X\prec Z\prec A\prec Y$.
The four columns correspond to
$\Delta_Z$, $\Delta_{\boldsymbol X}$, $\Delta_A$, and
$\Delta_{\mathrm{res}}$, and the rows correspond to the four combinations of
sample size $n$ and case-mix dimension $p$.
Boxplots summarize estimates across 500 Monte Carlo replicates.
The horizontal dashed line in each panel denotes the corresponding true
variance component.
The small colored point and interval show the Monte Carlo mean estimate and
its 95\% normal-approximation Monte Carlo error interval.
For the parametric main-effects, parametric interaction, and Firth outcome
models, results are shown using both direct modeling of
$Z\mid\boldsymbol X$ and Bayes inversion from a fitted model for
$\boldsymbol X\mid Z$; XGBoost uses the direct modified-order
implementation.
}
\label{fig:supp_modified_gamma0}
\end{figure}

\begin{figure}[!htbp]
\centering
\includegraphics[width=1.0\textwidth]
{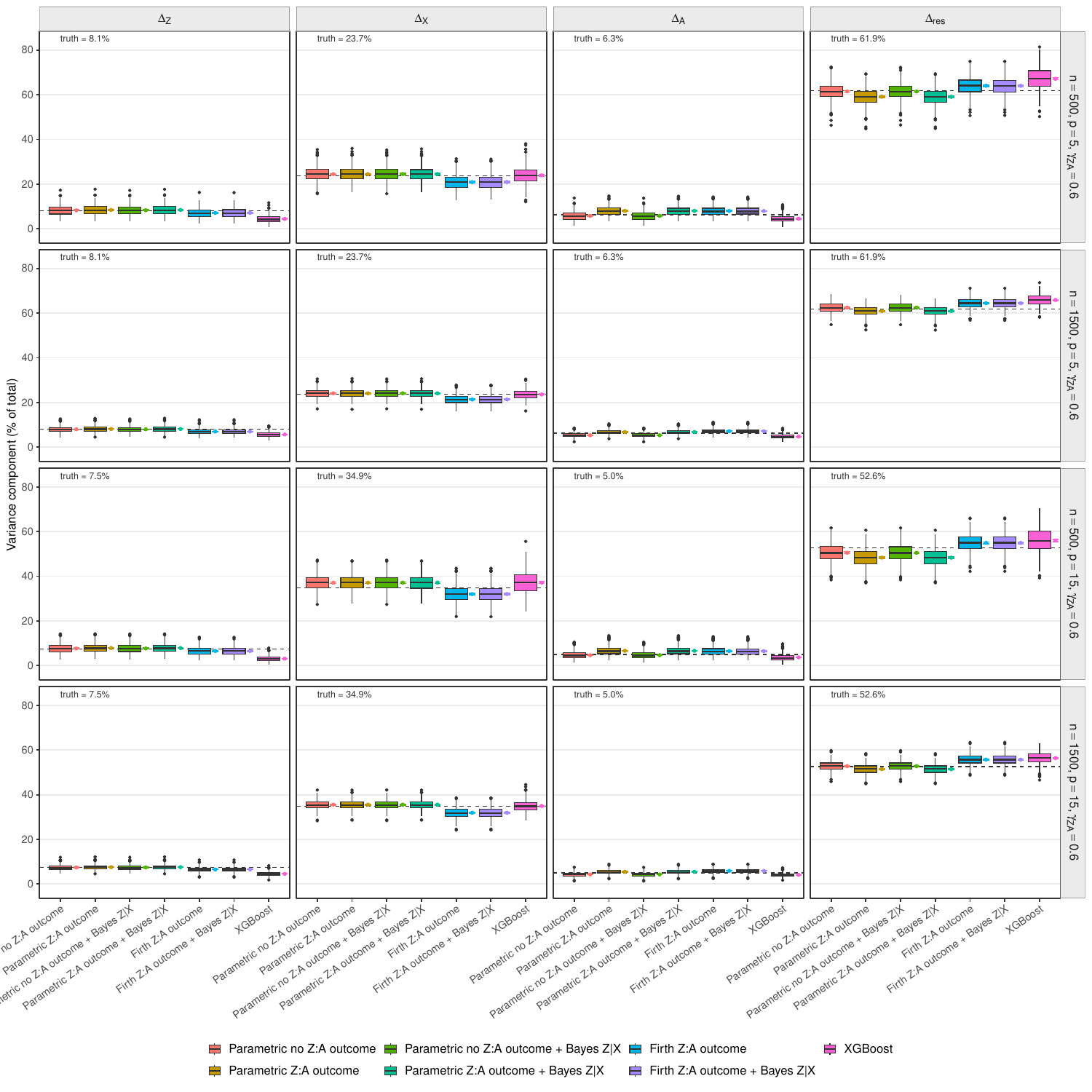}
\caption{
Finite-sample distributions of the estimated modified-order variance
components when $\gamma_{ZA}=0.6$, expressed as percentages of the total
outcome variance.
The modified ordering is
$\boldsymbol X\prec Z\prec A\prec Y$.
The four columns correspond to
$\Delta_Z$, $\Delta_{\boldsymbol X}$, $\Delta_A$, and
$\Delta_{\mathrm{res}}$, and the rows correspond to the four combinations of
sample size $n$ and case-mix dimension $p$.
Boxplots summarize estimates across 500 Monte Carlo replicates.
The horizontal dashed line in each panel denotes the corresponding true
variance component.
The small colored point and interval show the Monte Carlo mean estimate and
its 95\% normal-approximation Monte Carlo error interval.
For the parametric main-effects, parametric interaction, and Firth outcome
models, results are shown using both direct modeling of
$Z\mid\boldsymbol X$ and Bayes inversion from a fitted model for
$\boldsymbol X\mid Z$; XGBoost uses the direct modified-order
implementation.
}
\label{fig:supp_modified_gamma06}
\end{figure}

\begin{figure}[!htbp]
\centering
\includegraphics[width=1.0\textwidth]
{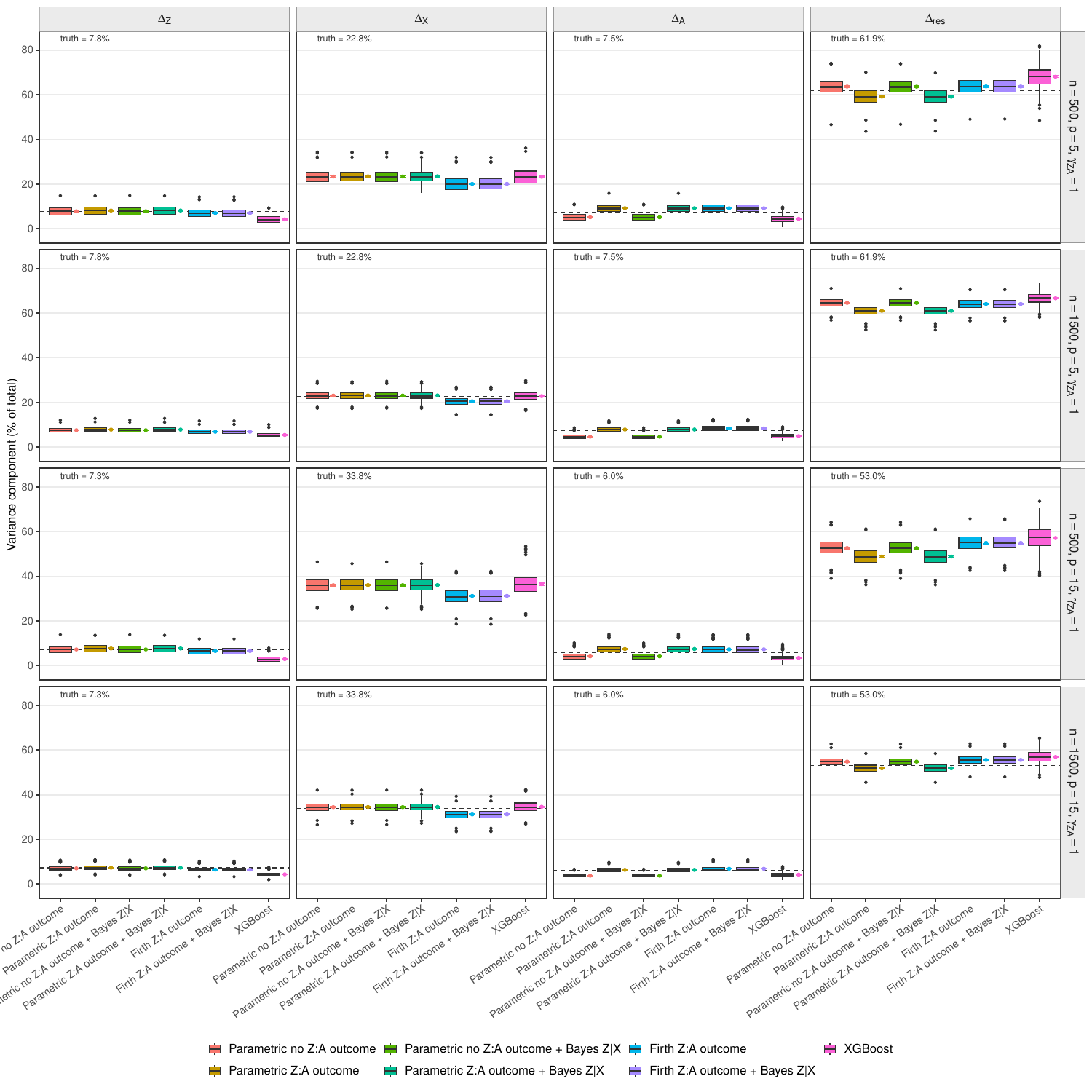}
\caption{
Finite-sample distributions of the estimated modified-order variance
components when $\gamma_{ZA}=1$, expressed as percentages of the total
outcome variance.
The modified ordering is
$\boldsymbol X\prec Z\prec A\prec Y$.
The four columns correspond to
$\Delta_Z$, $\Delta_{\boldsymbol X}$, $\Delta_A$, and
$\Delta_{\mathrm{res}}$, and the rows correspond to the four combinations of
sample size $n$ and case-mix dimension $p$.
Boxplots summarize estimates across 500 Monte Carlo replicates.
The horizontal dashed line in each panel denotes the corresponding true
variance component.
The small colored point and interval show the Monte Carlo mean estimate and
its 95\% normal-approximation Monte Carlo error interval.
The parametric main-effects outcome model omits the true $Z$-by-$A$
interaction, whereas the interaction and Firth models include it.
For each parametric outcome model, results are shown using both direct
modeling of $Z\mid\boldsymbol X$ and Bayes inversion from a fitted model for
$\boldsymbol X\mid Z$; XGBoost uses the direct modified-order
implementation.
}
\label{fig:supp_modified_gamma1}
\end{figure}

\begin{figure}[!htbp]
\centering
\includegraphics[width=0.95\textwidth]
{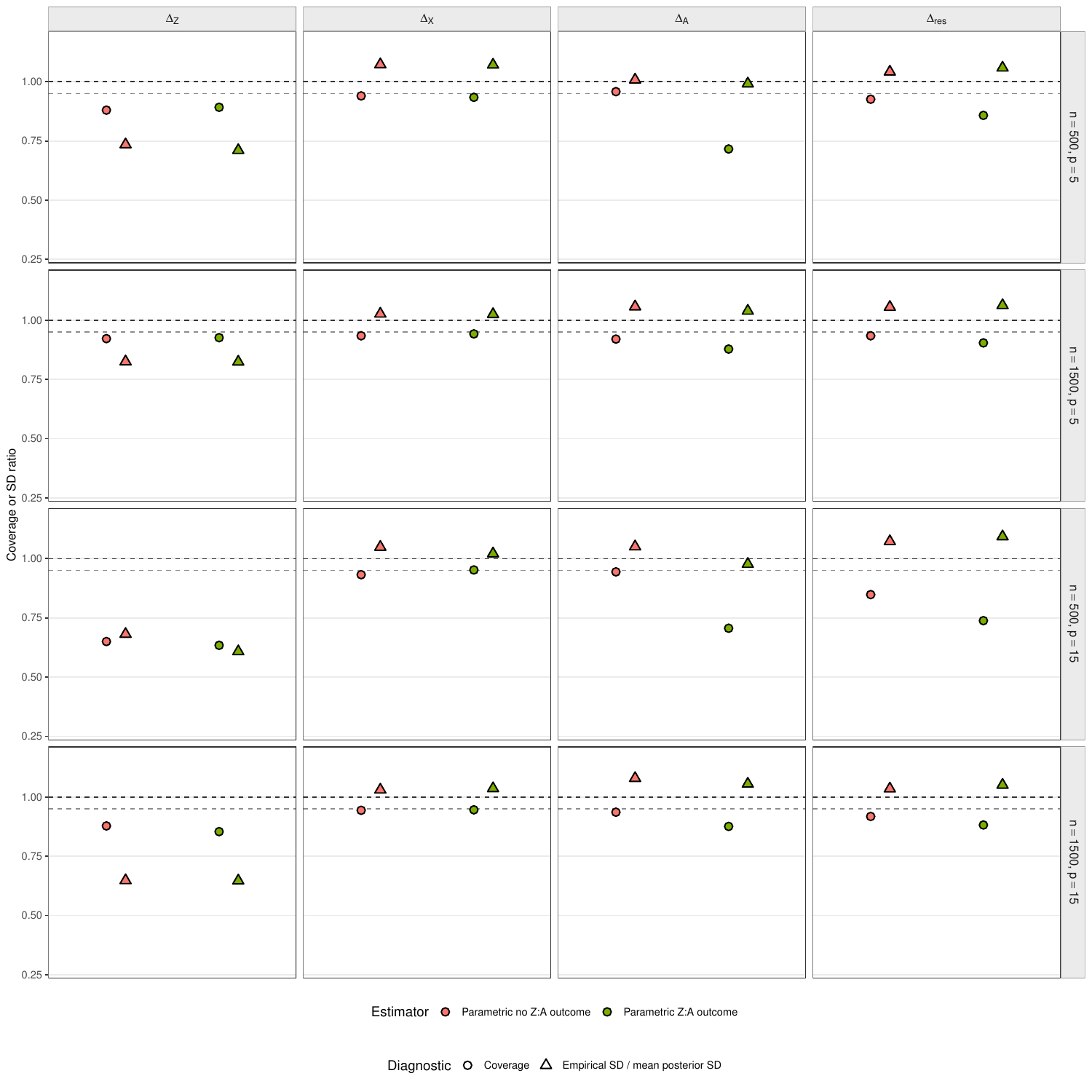}
\caption{
Uncertainty diagnostics for the topological-order decomposition when
$\gamma_{ZA}=0$.
The four columns correspond to the variance components
$\Delta_Z$, $\Delta_{\boldsymbol X}$, $\Delta_A$, and
$\Delta_{\mathrm{res}}$, and the rows correspond to the four combinations of
sample size $n$ and case-mix dimension $p$.
Circles show empirical coverage probabilities of the nominal 95\% approximate
Bayesian posterior intervals across 500 Monte Carlo replicates, and triangles
show the ratio of the empirical standard deviation of the point estimator to
the mean posterior standard deviation.
The horizontal dashed reference lines at 0.95 and 1 correspond to nominal
coverage and agreement between repeated-sampling and posterior standard
deviations, respectively.
Results are shown for the parametric main-effects outcome model and the
parametric outcome model including the $Z$-by-$A$ interaction.
}
\label{fig:supp_topological_uncertainty_gamma0}
\end{figure}

\begin{figure}[!htbp]
\centering
\includegraphics[width=0.95\textwidth]
{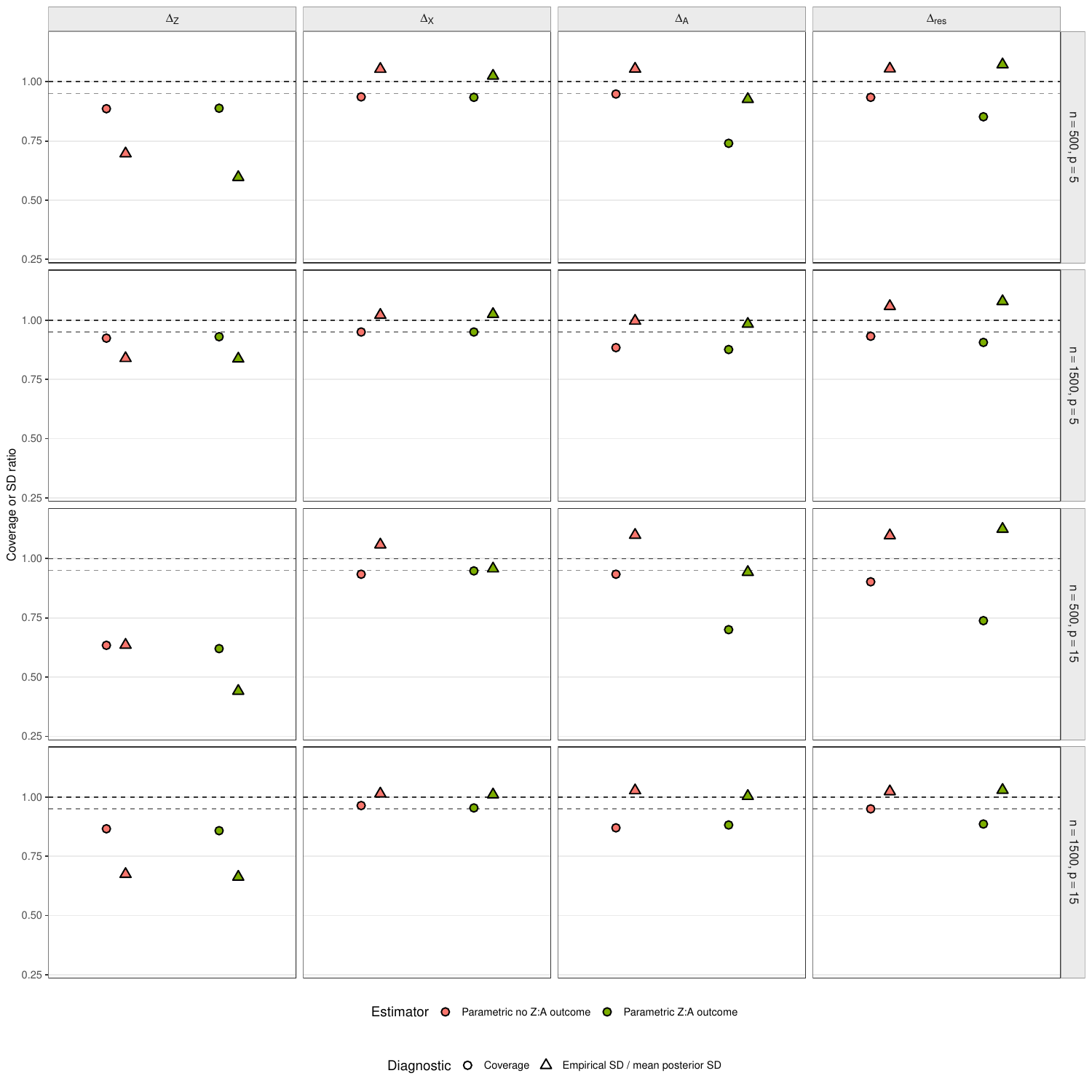}
\caption{
Uncertainty diagnostics for the topological-order decomposition when
$\gamma_{ZA}=0.6$.
The four columns correspond to
$\Delta_Z$, $\Delta_{\boldsymbol X}$, $\Delta_A$, and
$\Delta_{\mathrm{res}}$, and the rows correspond to the four combinations of
sample size $n$ and case-mix dimension $p$.
Circles show empirical coverage probabilities of the nominal 95\% approximate
Bayesian posterior intervals across 500 Monte Carlo replicates, and triangles
show the ratio of the empirical standard deviation of the point estimator to
the mean posterior standard deviation.
The horizontal dashed reference lines at 0.95 and 1 indicate nominal coverage
and agreement between repeated-sampling and posterior standard deviations,
respectively.
The parametric main-effects outcome model omits the true $Z$-by-$A$
interaction, whereas the interaction model includes it.
}
\label{fig:supp_topological_uncertainty_gamma06}
\end{figure}

\begin{figure}[!htbp]
\centering
\includegraphics[width=0.95\textwidth]
{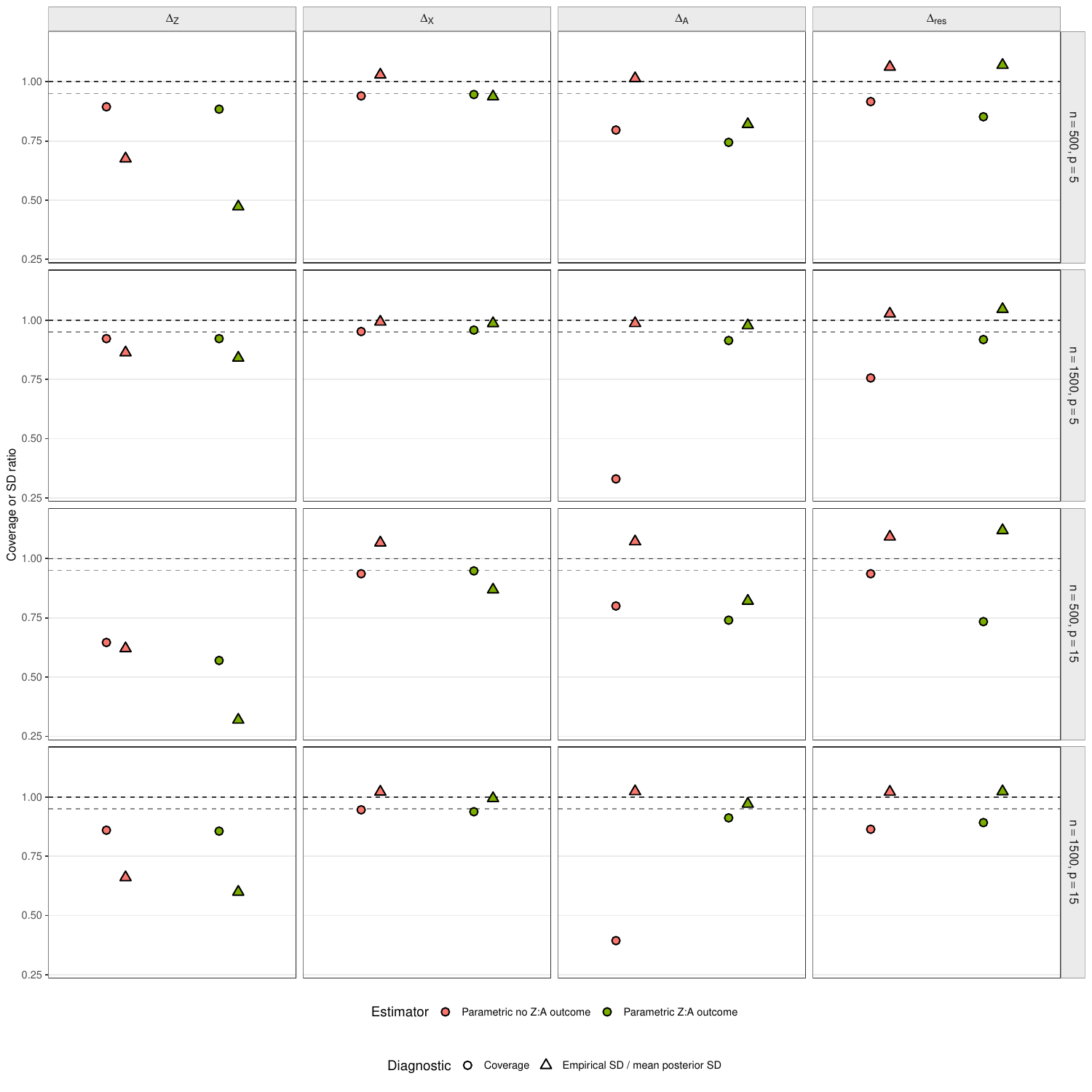}
\caption{
Uncertainty diagnostics for the topological-order decomposition when
$\gamma_{ZA}=1$.
The four columns correspond to
$\Delta_Z$, $\Delta_{\boldsymbol X}$, $\Delta_A$, and
$\Delta_{\mathrm{res}}$, and the rows correspond to the four combinations of
sample size $n$ and case-mix dimension $p$.
Circles show empirical coverage probabilities of the nominal 95\% approximate
Bayesian posterior intervals across 500 Monte Carlo replicates, and triangles
show the ratio of the empirical standard deviation of the point estimator to
the mean posterior standard deviation.
The horizontal dashed reference lines at 0.95 and 1 indicate nominal coverage
and agreement between repeated-sampling and posterior standard deviations,
respectively.
The parametric main-effects outcome model omits the true $Z$-by-$A$
interaction, whereas the interaction model includes it.
}
\label{fig:supp_topological_uncertainty_gamma1}
\end{figure}

\begin{figure}[!htbp]
\centering
\includegraphics[width=0.95\textwidth]
{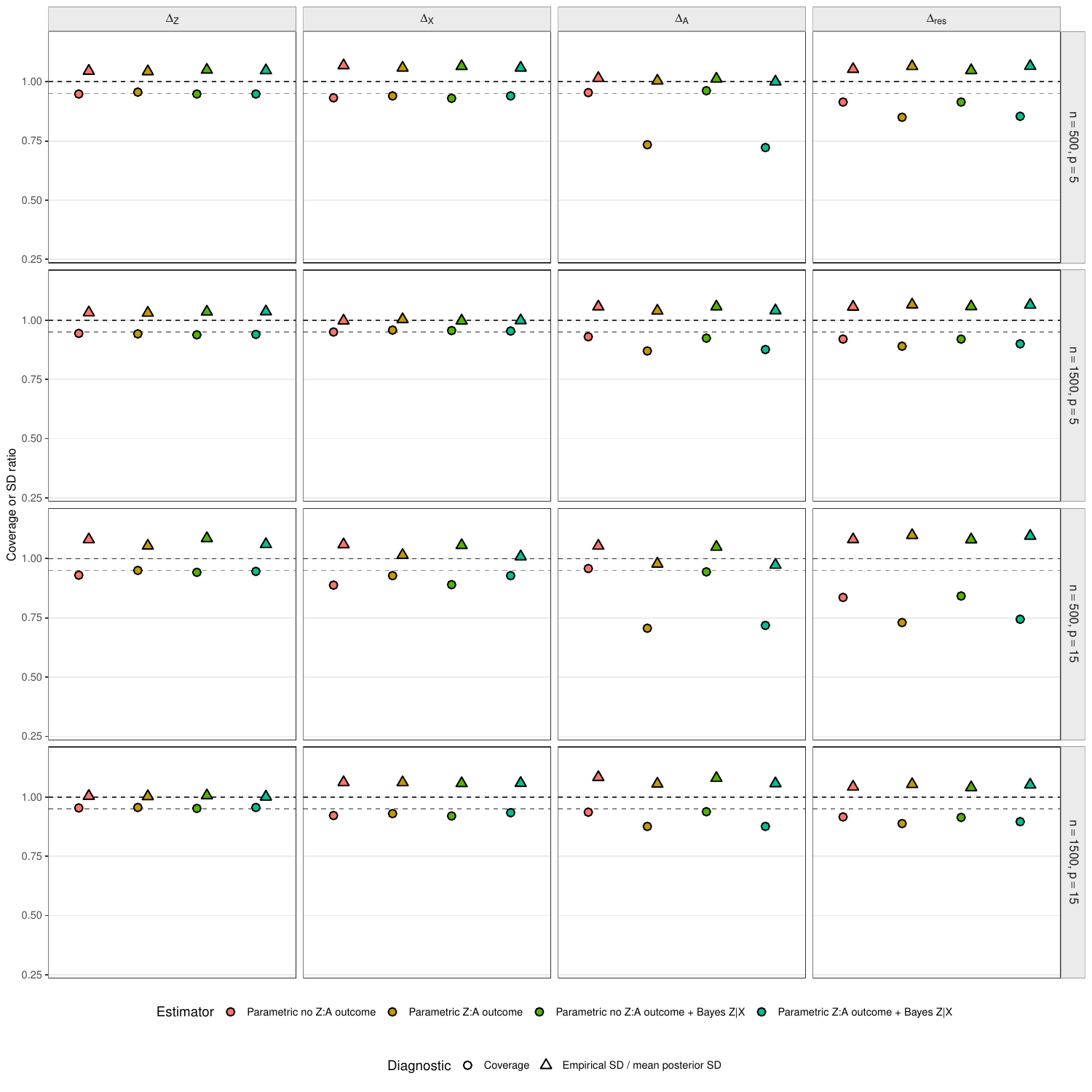}
\caption{
Uncertainty diagnostics for the modified-order decomposition
$\boldsymbol X\prec Z\prec A\prec Y$ when $\gamma_{ZA}=0$.
The four columns correspond to the variance components
$\Delta_Z$, $\Delta_{\boldsymbol X}$, $\Delta_A$, and
$\Delta_{\mathrm{res}}$, and the rows correspond to the four combinations of
sample size $n$ and case-mix dimension $p$.
Circles show empirical coverage probabilities of the nominal 95\% approximate
Bayesian posterior intervals across 500 Monte Carlo replicates, and triangles
show the ratio of the empirical standard deviation of the point estimator to
the mean posterior standard deviation.
The horizontal dashed reference lines at 0.95 and 1 indicate nominal coverage
and agreement between repeated-sampling and posterior standard deviations,
respectively.
For both the parametric main-effects and $Z$-by-$A$ interaction outcome
models, results are shown using direct modeling of $Z\mid\boldsymbol X$ and
using Bayes inversion from a fitted model for $\boldsymbol X\mid Z$.
}
\label{fig:supp_modified_uncertainty_gamma0}
\end{figure}

\begin{figure}[!htbp]
\centering
\includegraphics[width=0.95\textwidth]
{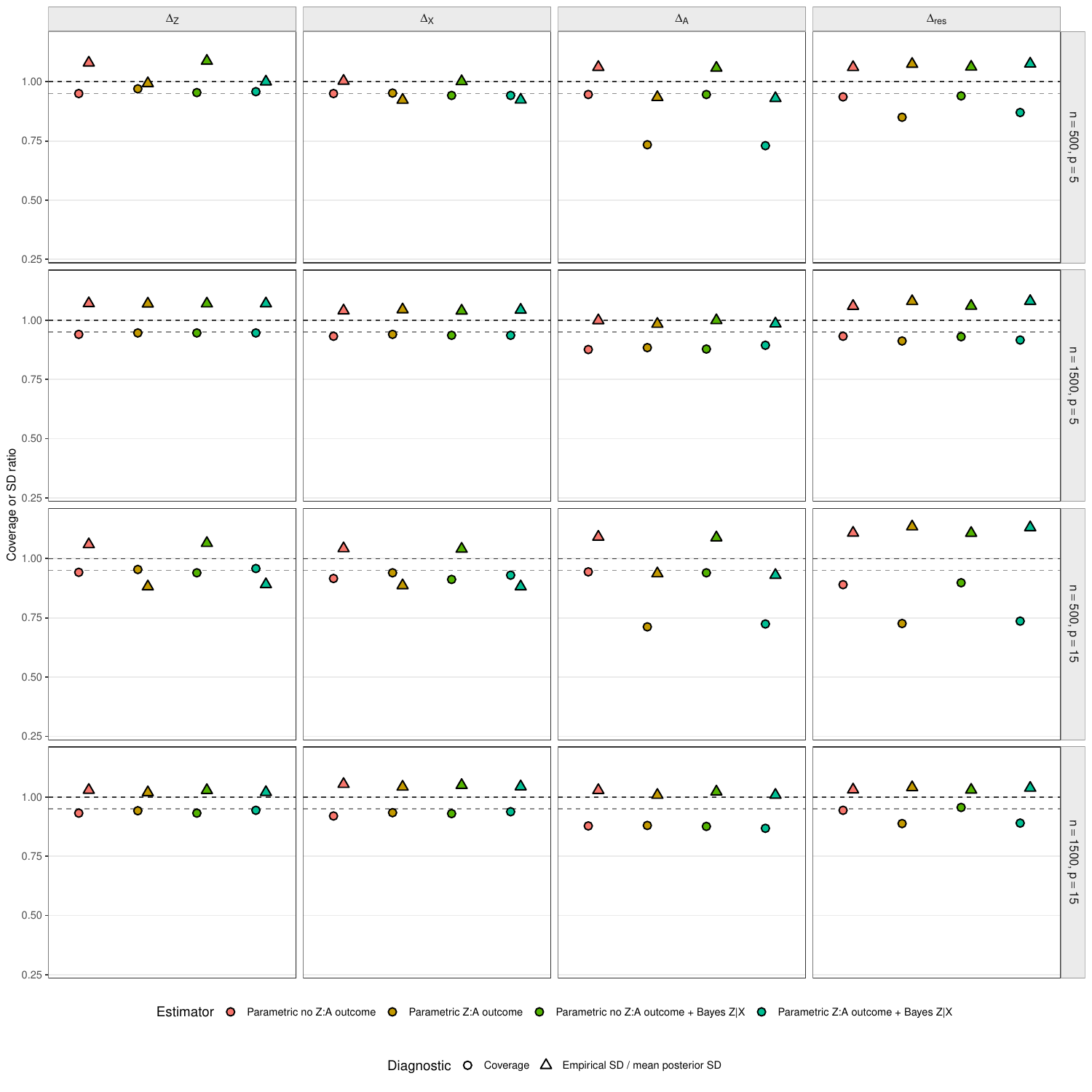}
\caption{
Uncertainty diagnostics for the modified-order decomposition
$\boldsymbol X\prec Z\prec A\prec Y$ when $\gamma_{ZA}=0.6$.
The four columns correspond to
$\Delta_Z$, $\Delta_{\boldsymbol X}$, $\Delta_A$, and
$\Delta_{\mathrm{res}}$, and the rows correspond to the four combinations of
sample size $n$ and case-mix dimension $p$.
Circles show empirical coverage probabilities of the nominal 95\% approximate
Bayesian posterior intervals across 500 Monte Carlo replicates, and triangles
show the ratio of the empirical standard deviation of the point estimator to
the mean posterior standard deviation.
The horizontal dashed reference lines at 0.95 and 1 indicate nominal coverage
and agreement between repeated-sampling and posterior standard deviations,
respectively.
For both outcome-model specifications, results are shown using direct
modeling of $Z\mid\boldsymbol X$ and Bayes inversion from a fitted model for
$\boldsymbol X\mid Z$.
The parametric main-effects outcome model omits the true $Z$-by-$A$
interaction, whereas the interaction model includes it.
}
\label{fig:supp_modified_uncertainty_gamma06}
\end{figure}

\begin{figure}[!htbp]
\centering
\includegraphics[width=0.95\textwidth]
{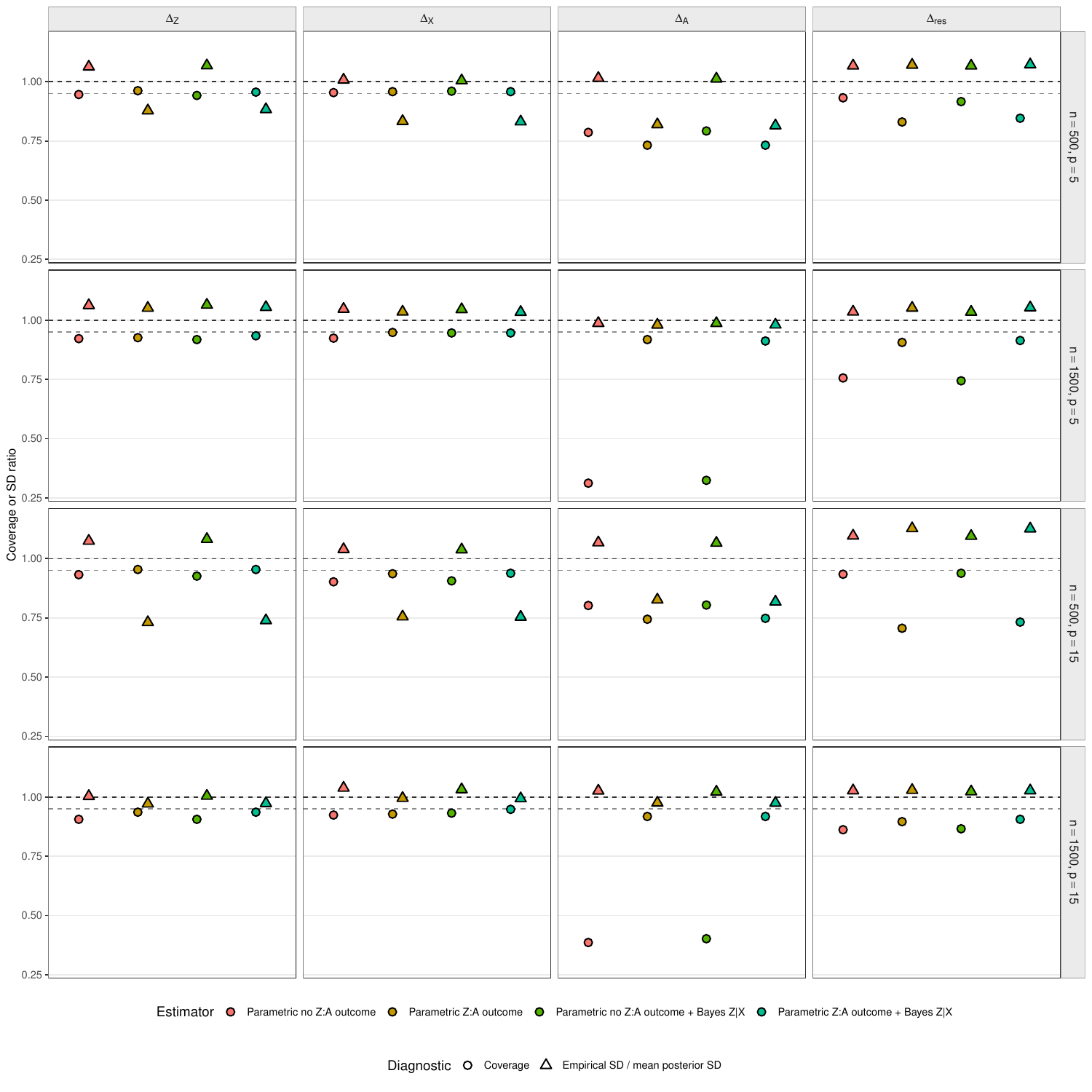}
\caption{
Uncertainty diagnostics for the modified-order decomposition
$\boldsymbol X\prec Z\prec A\prec Y$ when $\gamma_{ZA}=1$.
The four columns correspond to
$\Delta_Z$, $\Delta_{\boldsymbol X}$, $\Delta_A$, and
$\Delta_{\mathrm{res}}$, and the rows correspond to the four combinations of
sample size $n$ and case-mix dimension $p$.
Circles show empirical coverage probabilities of the nominal 95\% approximate
Bayesian posterior intervals across 500 Monte Carlo replicates, and triangles
show the ratio of the empirical standard deviation of the point estimator to
the mean posterior standard deviation.
The horizontal dashed reference lines at 0.95 and 1 indicate nominal coverage
and agreement between repeated-sampling and posterior standard deviations,
respectively.
For both outcome-model specifications, results are shown using direct
modeling of $Z\mid\boldsymbol X$ and Bayes inversion from a fitted model for
$\boldsymbol X\mid Z$.
The parametric main-effects outcome model omits the true $Z$-by-$A$
interaction, whereas the interaction model includes it.
}
\label{fig:supp_modified_uncertainty_gamma1}
\end{figure}

\end{document}